\documentclass[letterpaper, 11pt, notitlepage]{article}

\usepackage{graphicx, amsmath, fancyhdr, lastpage, bm, multirow, tikz, float, natbib, setspace, fullpage, amsthm, amsfonts, comment, caption, subcaption}
\usetikzlibrary{arrows}
\newtheorem{mythm}{Theorem}

\newtheorem{mylemma}{Lemma}

\begin{document}
\begin{doublespace}
\title{A Pivot-Based Improvement to Sandwich-Based Confidence Intervals
}
\author{James W. Harmon$^{1}$ and Peter D. Hoff$^{1,2}$ \\
Departments of Statistics$^{1}$ and Biostatistics$^{2}$ \\
University of Washington}
\date{\today}
\maketitle

\begin{abstract}
The current standard for confidence interval construction in the context of a possibly misspecified model is to use an interval based on the sandwich estimate of variance.  These intervals provide asymptotically correct coverage, but small-sample coverage is known to be poor.  By eliminating a plug-in assumption, we derive a pivot-based method for confidence interval construction under possibly misspecified models.  When compared against confidence intervals generated by the sandwich estimate of variance, this method provides more accurate coverage of the pseudo-true parameter at small sample sizes.  This is shown in the results of several simulation studies.  Asymptotic results show that our pivot-based intervals have large sample efficiency equal to that of intervals based on the sandwich estimate of variance.  
\end{abstract}

\section{Introduction}

The method of maximum likelihood is a powerful and efficient technique for statistical inference, but much of the efficiency is lost if the data-generating model is incorrectly specified.  If a statistician selects the wrong family of distributions to model the data, then the standard variance estimates using the Fisher information matrix will likely be biased, even asymptotically.  Consequently, any inference performed with the incorrect model will be adversely affected.  An example of this is when true coverage rates of confidence intervals do not match nominal coverage rates even as the sample size increases.

This begs the question of why any analyst would want to use maximum likelihood or any parametric methods at all.  One answer is that the parameter in a misspecified model might still be the appropriate parameter to estimate.  An example of this might be the measure of linear association between a covariate and the response variable in a regression model.  Even if the full model has been misspecified, this association might still be of interest to a researcher.  Given certain regularity conditions, the maximum likelihood estimator converges to a unique pseudo-true parameter.  This pseudo-true parameter has a nice interpretation:  it minimizes the Kullback-Leibler distance between the true data-generating model and assumed family of models \citep{akaike1973information}.  In other words, the pseudo-true parameter yields the distribution from the misspecified model that is ``closest'' to the true data-generating distribution.  The pseudo-true parameter might well be of interest, even if the model is misspecified.  For example, an analyst might be interested in knowing the linear association between a covariate and the response in a linear regression.

The second answer is that parametric methods increase inferential efficiency, even if the model is incorrect.  As one moves through the spectrum of parametric, semi-parametric, and non-parametric methods, assumptions about the data-generating model decrease.  The trade-off for fewer assumptions is that the efficiency of one's inference decreases as well.  In terms of confidence intervals, decreased efficiency means larger confidence intervals for a given confidence level and sample size.  Using fully empirical methods clearly requires the fewest assumptions of the true data-generating model, but the cost in efficiency might be too high.  Once the decision to use parametric or semi-parametric methods has been made, the question then boils down to one of how to perform inference about the pseudo-true parameter.  

The standard method for creating confidence intervals using misspecified models was introduced by \citet{huber1967behavior} and expanded by \citet{white1982maximum}.  It is based on what is commonly called the ``sandwich'' estimate of variance, so named because the matrix version of this estimator has the form $A^{-1} B A^{-1}$ where the outer matrices are the ``bread'' and the inner matrix is the ``peanut butter''.  The main benefit of the sandwich estimate of variance is that it provides asymptotically correct variance estimates for maximum likelihood parameter estimates.  This yields confidence intervals with asymptotically correct coverage rates for pseudo-true parameters.  Additionally, it is applicable across a wide range of statistical techniques and is not limited to maximum likelihood methods \citep{chen2005pseudo}, \citep{huber2011robust}.

While the sandwich estimate of variance provides asymptotically correct inference, the small sample confidence interval coverage is typically quite poor.  This has been attributed to the increased variability in the sandwich estimator over other, less robust variance estimators \citep{kauermann2001note}.  Specifically, \citet{efron1986discussion} in his discussion of \citet{wu1986jackknife} noticed greater variability in the sandwich estimator than in the standard variance estimator in his simulation results.  \citet{breslow1990tests} also finds greater variability in the sandwich estimator in his simulations of overdispersed Poisson regression.  \citet{firth1992discussion} and \citet{mccullagh1992discussion} in their review of \citet{liang1992multivariate} both question the small sample efficiency of the sandwich.  And \citet{diggle2009analysis} note that best results are obtained when there are ``many experimental units''.  The result of greater variability in typical small sample scenarios is that confidence intervals may only yield much less than nominal coverage of the pseudo-true parameter.  Greater variability and consequently poor confidence interval coverage is the result both of multiple approximations being used in the derivation of confidence intervals based on the sandwich estimate of variance and of multiple quantities being estimated.  

Several authors have attempted to adjust the sandwich estimate of variance to improve confidence interval coverage in small sample sizes.  The simplest, though not the first, adjustment to the sandwich is to multiply the standard estimator by $\tfrac{n}{n-p}$ where $p$ is the number of parameters in the model and $n$ is the size of the dataset \citep{hinkley1977jackknifing}.  Following the convention laid out in \citet{hardin2003sandwich}, we call this adjustment ``hc2''.  Another earlier attempt to correct for the bias of the sandwich's variance estimate was \citet{horn1975estimating}.  This correction divided the $i^{th}$ diagonal element in the inner, ``peanut butter'' matrix by $1-h_{ii}$, where $h_{ii}$ is the $i^{th}$ diagonal element from the hat matrix.  The hat matrix refers to the projection matrix generated by the design matrix in linear regression.  We call this adjustment ``hc1''.  A jackknife estimator was used by \citet{efron1982jackknife}, which is equivalent to dividing the $i^{th}$ diagonal element in the inner, ``peanut butter'' matrix by $(1-h_{ii})^2$.  We call this adjustment ``hc3''.  In \citet{kauermann2000sandwich} an adjusted normal quantile value is computed, instead of adjusting the sandwich directly.  We call this adjustment ``hc4''.  \citet{fay2001small} compute the degrees of freedom needed for approximately correct $t-$ and $F-$distributions, instead of using asymptotically correct normal and chi-squared distributions.  We call this adjustment ``hc5''.  

As remarked upon in \citet{hardin2003sandwich}, all of these adjustments to the sandwich variance estimator are ``ad hoc solutions''. We think a better critique would be that they are ``post hoc solutions'', since they all attempt to adjust the sandwich estimate of variance after the fact.  All of them attempt to inflate confidence interval sizes, but none of them do so in a way that directly addresses the variability inherent in the sandwich estimate of variance.  Two of the methods (\citep{kauermann2000sandwich} and \citep{fay2001small}) do not focus on adjusting the sandwich estimate of variance.  Instead, they modify the standard normal quantiles used in confidence interval construction, or they replace asymptotic normal and chi-squared distributions with Student-$t$ and $F$ distributions.  To date, there is no universally-agreed-upon adjustment or alternative to the sandwich variance estimator for small sample sizes that corrects for the its structural problems.  

Instead of trying to adjust the sandwich estimate of variance after the fact, we analyze the argument that produces the sandwich.  We find three approximating assumptions in this argument.  By eliminating one of these approximating assumptions and remembering that the ultimate goal is inference on a parameter, we derive a pivot-based method of inference that produces more accurate confidence intervals in small sample sizes than the sandwich estimate of variance and its ``post hoc'' adjustments.  

The rest of the paper is organized as follows.  In Section \ref{method}, we derive our pivot-based method by minimizing the number of assumptions used in the standard sandwich-based method.  By avoiding plug-in assumptions commonly used with sandwich-based inference, we preserve any mean-variance relationships in the model and improve confidence interval coverage in smaller sample sizes.  In Section \ref{oneparam}, we give a proof in the one-parameter case that no asymptotic efficiency is lost by using the pivot versus the sandwich.  In that same section we follow up with simulation examples that show we can achieve improved coverage in small sample sizes.  In Section \ref{multparam}, we extend our proof of asymptotic equivalence to the multiple parameter case, and show a corresponding simulation exercise using a simple linear regression model.  Section \ref{realworld} takes data from a real-world example, runs a multiple regression on the data, and computes confidence region coverage for all parameters combined.  

\section{Pivot Method}\label{method}
The argument deriving the sandwich estimate of variance makes three approximating assumptions.  In this Section we review those assumptions.  Our approach to inference eliminates the plug-in assumption, resulting in better confidence interval coverage at small sample sizes.

We now describe the derivation of the sandwich estimate of variance for a one-parameter model under misspecification.  Let $y_1, \ldots, y_n$ represent data values.  Let the family of distributions assumed by the analyst be described by the set of densities $\{ f(y;\theta): \theta \in \Theta \}$ with indexing parameter $\theta$, and $\Theta$ being a compact subset of $R^p$.  We will write true distribution of the data as density $g(y)$, where $g(y)$ is not necessarily in $\{f(y;\theta) : \theta \in \Theta \}$.  Let $\prod_{i=1}^n f(y_i;\theta)$ be the pseudo-likelihood of independent and identically distributed data.  Then the pseudo-log-likelihood is given by $\sum_{i=1}^n \log(f(y_i;\theta))$.  We denote the derivative of the pseudo-log-likelihood with respect to $\theta$ by $\sum_{i=1}^n l_{\theta}(y_i; \theta)$ where the subscript denotes the derivative. The second derivative of the pseudo-log-likelihood with respect to $\theta$ is similarly denoted by $\sum_{i=1}^n l_{\theta \theta}(y_i; \theta)$.  We let $\hat{\theta}$ denote the maximum misspecified likelihood estimator.  This is simply the estimator we compute when using maximum likelihood methods on a possibly misspecified likelihood.  We use $\theta^{\ast}$ to represent the Kullback-Leibler minimizing parameter, which we refer to as the pseudo-true parameter. Then by a mean value lemma for random variables proven in \citet{jennrich1969asymptotic} there exists $\bar{\theta}$, a weighted average of $\hat{\theta}$ and $\theta^{\ast}$, that satisfies the following equation
\begin{align} \label{mainjennrich}
- \frac{1}{\sqrt{n}} \sum_{i=1}^n l_{\theta}(y_i, \theta^{\ast}) = \frac{1}{\sqrt{n}} \left[ \sum_{i=1}^n l_{\theta \theta}(y_i, \bar{\theta}) \right] (\hat{\theta} - \theta^{\ast})
\end{align}
From here we solve for $\sqrt{n}(\hat{\theta} - \theta^{\ast})$ and compute the variance of its limiting distribution.  In the process, one must make three approximating assumptions:  
\begin{enumerate}
\item $\tfrac{1}{\sqrt{n}} \sum_{i=1}^n l_{\theta}(y_i, \theta^{\ast}) \overset{\cdot}{\sim} N(0,B(\theta^{\ast}))$
\item $\tfrac{1}{n} \sum_{i=1}^n [l_{\theta}(y_i, \theta^{\ast})]^2 = B_n(\theta^{\ast}) \approx E\{[l_{\theta}(y,\theta^{\ast})]^2\} = B(\theta^{\ast})$
\item $\hat{\theta} \approx \theta^{\ast}$.
\end{enumerate}
Assumption 1 is a distributional assumption.  Since we assume a possibly misspecified model, we cannot state the exact distribution of the score function.  Instead, we rely on the Central Limit Theorems for an approximate distribution.  Assumption 2 is an empirical assumption.  Again because we assume a possibly misspecified model we cannot compute the expectation in that line.  We use the empirical approximation instead.  Assumptions 3 is the plug-in assumption.  We do not know the true value of the pseudo-true parameter, so we approximate the value with its estimator.  

If we let $A_n(\bar{\theta}) = \tfrac{1}{n} \sum l_{\theta \theta}(y_i, \bar{\theta})$ then the argument proceeds as follows:
\begin{align}
\frac{1}{\sqrt{n}} \sum l_{\theta}(y_i, \theta^{\ast}) &\overset{\cdot}{\sim} N(0,B(\theta^{\ast})) \label{inferenceline}
\end{align}
is directly from assumption (i).
\begin{align}
A_n(\bar{\theta}) \sqrt{n}(\hat{\theta} - \theta^{\ast}) &\overset{\cdot}{\sim} N(0,B(\theta^{\ast})) \label{line2}
\end{align}
is derived by replacing the score function in assumption (i) with the right-hand side of (\ref{mainjennrich}).
\begin{align}
B^{-1/2}(\theta^{\ast}) A_n(\bar{\theta}) \sqrt{n}(\hat{\theta} - \theta^{\ast}) &\overset{\cdot}{\sim} N(0,1) \label{line3}
\end{align}
is computed by multiplying both sides by $B^{-1/2}(\theta^{\ast})$.
\begin{align}
B_n^{-1/2}(\theta^{\ast}) A_n(\bar{\theta}) \sqrt{n}(\hat{\theta} - \theta^{\ast}) &\overset{\cdot}{\sim} N(0,1) \label{line4} 
\end{align}
results from applying assumption (ii).
\begin{align}
B_n^{-1/2}(\hat{\theta}) A_n(\hat{\theta}) \sqrt{n}(\hat{\theta} - \theta^{\ast}) &\overset{\cdot}{\sim} N(0,1) \label{line5}
\end{align}
uses assumptions (iii) and (iv), replacing $\theta^{\ast}$ with $\hat{\theta}$ in $A_n$.  Recall that $\bar{\theta}$ is a linear combination of $\hat{\theta}$ and $\theta^{\ast}$, so replacing $\bar{\theta}$ with $\hat{\theta}$ in $B_n^{-1/2}$ is equivalent to replacing $\theta^{\ast}$ with $\hat{\theta}$.  

Our goal is to eliminate one or more of these approximating assumptions in order to improve accuracy in confidence intervals about $\theta^{\ast}$.  In our derivation we also start with line (\ref{inferenceline}), but we remind ourselves that finding the asymptotic distribution of $\sqrt{n}(\hat{\theta} - \theta^{\ast})$ is only a tool for creating confidence intervals and is not our end goal.  Our goal is to improve confidence interval coverage about $\theta^{\ast}$.  

Since we are operating under the assumption of model misspecification, we have no knowledge of the exact distribution of $\tfrac{1}{\sqrt{n}} \sum_{i=1}^n l_{\theta}(y_i, \theta^{\ast})$.  The best we can do is to keep the assumption of approximate normality by appealing to the Central Limit Theorem and using assumption (i) above.  This may be a crude approximation, but it is the only distributional result we have.  Model misspecification also means we cannot compute the asymptotic variance $B(\theta)$, so we estimate it with its empirical approximation, $B_n(\theta)$, which is assumption (ii) from above.  With these assumptions alone, we can still perform inference using an approximate pivotal quantity.  The pivot and its asymptotically correct distribution is derived now.
\begin{align}
\frac{1}{\sqrt{n}} \sum l_{\theta}(y_i, \theta^{\ast}) &\overset{\cdot}{\sim} N(0,B(\theta^{\ast})) 
\end{align}
is the same as line (\ref{inferenceline}) and 
\begin{align}
B^{-1/2}(\theta^{\ast}) \frac{1}{\sqrt{n}} \sum l_{\theta}(y_i, \theta^{\ast}) &\overset{\cdot}{\sim} N(0,1) 
\end{align}
is the result of multiplying both sides by $B^{-1/2}(\theta^{\ast})$.
Applying Assumption 2 we get
\begin{align}
B_n^{-1/2}(\theta^{\ast}) \frac{1}{\sqrt{n}} \sum l_{\theta}(y_i, \theta^{\ast}) &\overset{\cdot}{\sim} N(0,1) 
\end{align}
And the explicit, approximate pivot is
\begin{align}
t(y_1, \ldots, y_n; \theta^{\ast}) &= \left\{ \frac{1}{n} \sum_{i=1}^n [l_{\theta}(y_i, \theta^{\ast})]^2 \right\}^{-1/2} \frac{1}{\sqrt{n}} \sum_{i=1}^n l_{\theta}(y_i, \theta^{\ast}) \overset{\cdot}{\sim} N(0,1). \label{mainpivot}
\end{align}

With (\ref{mainpivot}), one can perform a linear search on $\theta^{\ast}$ to find confidence interval boundaries given a confidence level of $100(1-\alpha)\%$.  Find the $\theta$-values that makes the left-hand side of line (\ref{mainpivot}) equal to the desired standard normal quantiles, and that is the confidence interval.  Using this method, we have eliminated the plug-in approximating assumption required to perform inference using the sandwich estimate of variance. 

\section{One-Parameter Asymptotics and Simulation Example}\label{oneparam}
In this section we will do two things, both in the univariate parameter case.  First we show that the pivot loses no efficiency with respect to the sandwich.  This is important because if the pivot were only better than the sandwich at small sample sizes, but lost efficiency with respect to the sandwich at larger sample sizes, we would have no way of determining which method to use in order to perform the best inference possible.  Second, we give a couple of simulation examples comparing the pivot method with the sandwich and its post hoc adjustments.  Our examples show that the pivot method performs favorably.

\subsection{One-Parameter Asymptotics} \label{oneparamasy}
We now show that the confidence interval procedure based on pivot (\ref{mainpivot}) is asymptotically as efficient as the confidence interval based on the plug-in sandwich estimator.  In the context of confidence interval evaluation, efficiency is judged based on relative confidence interval width given a particular confidence level.  For a one-sided confidence interval, the equivalent comparison is of the values of the finite boundary.  We show in this subsection that
\begin{align} \label{uniasyequ}
\sqrt{n} | \theta_{n,1} - \theta_{n,2} | &\rightarrow_{a.s.} 0
\end{align}
where $\theta_{n,1}$ and $\theta_{n,2}$ that are defined by the equalities
\begin{align}
\sqrt{n} B_n^{-1/2}(\theta_{n,1}) \frac{1}{n} \sum l_{\theta}(y_i, \theta_{n,2}) &= z_{1-\alpha} \label{pivotmethod} \\
\sqrt{n} B_n(\hat{\theta})^{-1/2} A_n(\hat{\theta}) (\hat{\theta} - \theta_{n,2}) &= z_{1-\alpha}, \label{sandwichmethod}
\end{align}
which yield the boundaries of the one-sided confidence interval using the pivot, (\ref{pivotmethod}), and the sandwich, (\ref{sandwichmethod}), respectively.  To find a two-sided interval, we would set each equation above to $z_{1-\alpha/2}$ and $z_{\alpha/2}$, resulting in four total equations, and then solve for the interval endpoints.

In order to facilitate computation, we will substitute Equation (\ref{mainjennrich}) into Equation (\ref{pivotmethod}) to get
\begin{align}
\sqrt{n} B_n^{-1/2}(\theta_{n,1}) \frac{1}{n} \sum l_{\theta}(y_i, \theta_{n,2}) &= \sqrt{n} B_n(\theta_{n,1})^{-1/2} A_n(\bar{\theta}) (\hat{\theta} - \theta_{n,1}) = z_{1-\alpha} . \label{practicalpivot}
\end{align}

The reason for this approach is intuitive.  If two estimators have asymptotically normal distributions, then computing the relative efficiency of the estimators is equivalent to comparing confidence interval widths.  If the widths are the same, then the estimators are equally efficient.  If the widths are different, then the shorter confidence interval is more efficient.  

Expressions that we will use for each of $\theta_{n,1}$ and $\theta_{n,2}$ are derived from Equations (\ref{sandwichmethod}) and (\ref{practicalpivot})
\begin{align}
\theta_{n,1} &= \hat{\theta}_n - \frac{A_n((1-b_n)\hat{\theta}_n + b_n\theta_{1,n})^{-1} B_n(\theta_{n,1})^{1/2} z_{1-\alpha}}{\sqrt{n}} \label{pivplug}\\
\theta_{n,2} &= \hat{\theta}_n  - \frac{A_n(\hat{\theta})^{-1} B_n(\hat{\theta})^{1/2} z_{1-\alpha}}{\sqrt{n}} \label{sandplug}
\end{align}
where $b_n$ is the value between $0$ and $1$ that satisfies the Lemma 3 equality in \citet{jennrich1969asymptotic}.  Notice that we did not solve for $\theta_{n,1}$ explicitly.  This is not a problem.  We are considering asymptotic behavior, and we prove in the appendices that $\theta_{1,n} \rightarrow_{a.s.} \theta^{\ast}$.  Now consider the scaled difference
\begin{align}
\sqrt{n} | \theta_{n,1} - \theta_{n,2} | &= \sqrt{n} \left|\hat{\theta}_n - \frac{A_n((1-b_n)\hat{\theta}_n + b_n\theta_{n,1})^{-1} B_n(\theta_{n,1})^{1/2} z_{1-\alpha}}{\sqrt{n}} \right. \nonumber \\
 &- \left. \left(\hat{\theta}_n  - \frac{A_n(\hat{\theta})^{-1} B_n(\hat{\theta})^{1/2} z_{1-\alpha}}{\sqrt{n}}\right) \right| \label{convplug}\\
\end{align}
where we replace the left-hand side with the equivalent expressions from (\ref{pivplug}) and (\ref{sandplug}).  With some algebra we get
\begin{align}
 &=\sqrt{n} \left|\frac{A_n(\hat{\theta})^{-1} B_n(\hat{\theta})^{1/2} z_{1-\alpha}}{\sqrt{n}}  - \frac{A_n((1-b_n)\hat{\theta}_n + b_n\theta_{n,1})^{-1} B_n(\theta_{n,1})^{1/2} z_{1-\alpha}}{\sqrt{n}}\right| \label{convalg1}
\end{align}
and
\begin{align}
 &= \left|A_n(\hat{\theta})^{-1} B_n(\hat{\theta})^{1/2} z_{1-\alpha}  - A_n((1-b_n)\hat{\theta}_n + b_n\theta_{1,n})^{-1} B_n(\theta_{n,1})^{1/2} z_{1-\alpha} \right| . \label{convalg2}\\
\end{align}
And finally we have
\begin{align}
\left|A_n(\hat{\theta})^{-1} B_n(\hat{\theta})^{1/2} z_{1-\alpha}  - A_n((1-b_n)\hat{\theta}_n + b_n\theta_{1,n})^{-1} B_n(\theta_{n,1})^{1/2} z_{1-\alpha} \right| \rightarrow_{a.s.} \\
z_{1-\alpha} \left| A(\theta^{\ast})^{-1} B(\theta^{\ast})^{1/2} - A(\theta^{\ast})^{-1} B(\theta^{\ast})^{1/2} \right| \label{convergence}\\
 = 0.
\end{align}
These last lines are the consequence of several convergence results, the first being $\theta_{n,1} \rightarrow_{a.s.} \theta^{\ast}$, which is shown in the appendices.  By the strong law of large numbers, we have $A_n(\hat{\theta}) \rightarrow_{a.s.} A(\theta^{\ast})$, $A_n(\bar{\theta}) \rightarrow_{a.s.} A(\theta^{\ast})$, $B_n(\hat{\theta}) \rightarrow_{a.s.} B(\theta^{\ast})$, $B_n(\theta_{n,1}) \rightarrow_{a.s.} B(\theta^{\ast})$, and by Mann-Wald's Theorem we have line (\ref{convergence}) above.  Since the $\sqrt{n}$-distance between confidence interval boundaries converges to zero, the confidence intervals generated by the pivot and the sandwich are asymptotically equivalent.  Hence, the two procedures for generating confidence intervals are asymptotically equally efficient.

\subsection{One-Parameter Small Sample Comparisons}
In this section we examine two, one-parameter examples comparing confidence interval performance between intervals created using the likelihood pivot and intervals created using the sandwich estimate of variance.  The first example assumes a Poisson likelihood for count data.  The second example assumes a regression through the origin model with one response variable and one covariate.  Our choice of examples is meant to showcase the diverse applications of the likelihood pivot.  In our two examples we show that our method works with count data, with continuous data, in the presence of heteroscedasticity, and in a regression setting.  Broad applicability is one of the strengths of the likelihood pivot method for generating confidence intervals.  

\subsection{Poisson Data} \label{poisdat}
In this example, we assume a working model
\begin{align}
y_i &\overset{i.i.d.}{\sim} Pois(\theta) \qquad i=1,\ldots,n
\end{align}
and we wish to perform inference on the parameter of interest $\theta$.  A typical misspecification would be assuming the above model when an overdispersed Poisson is the truth.  We can generate such a distribution using the Negative Binomial distribution, since the variance of a Negative Binomial random variable can differ from its mean.  The Negative Binomial is what we used as our true data-generating distribution in this example.

In this model we have the following pieces of Equations \eqref{mainjennrich} and \eqref{mainpivot}
\begin{align}
l_{\theta}(\theta^{\ast}) &= \frac{1}{\theta^{\ast}} \sum_{i=1}^n y_i - n \\
\hat{\theta} &= \frac{1}{n} \sum_{i=1}^n y_i \\
l_{\theta \theta}(\bar{\theta}) &= \frac{-1}{(\bar{\theta})^2} \sum_{i=1}^n y_i
\end{align}
We then plug those into the pivot expression to get
\begin{align}
\frac{1}{\sqrt{n}} B_n(\theta^{\ast})^{-1/2} l_{\theta}(y_1, \ldots, y_n; \theta^{\ast}) &= \frac{1}{\sqrt{n}} \frac{\frac{1}{\theta^{\ast}} \sum_{i=1}^n y_i - n}{\sqrt{\frac{1}{(\theta^{\ast})^2} \frac{1}{n}\sum_{i=1}^n (y_i - \theta^{\ast})^2}} \\
 &= \sqrt{n} \frac{\frac{1}{\theta^{\ast}} \frac{1}{n} \sum_{i=1}^n y_i - 1}{\frac{1}{|\theta^{\ast}|}\sqrt{\frac{1}{n}\sum_{i=1}^n (y_i - \theta^{\ast})^2}} \\
 &= \sqrt{n} \frac{\frac{1}{n} \sum_{i=1}^n y_i - \theta^{\ast}}{\sqrt{\frac{1}{n}\sum_{i=1}^n (y_i - \theta^{\ast})^2}} \label{poispivot}
\end{align}
and we perform a linear search on the parameter in this pivot to find our confidence interval.

The sandwich estimate of variance is calculated as
\begin{align}
A_n(\hat{\theta})^{-1} B_n(\hat{\theta}) A_n(\hat{\theta})^{-1} &= \left[ \frac{-1}{n} \sum_{i=1}^n \frac{y_i}{\hat{\theta}^2} \right]^{-1} \frac{1}{n}\sum_{i=1}^n \left(\frac{y_i}{\hat{\theta}} - 1 \right)^2 \left[ \frac{-1}{n} \sum_{i=1}^n \frac{y_i}{\hat{\theta}^2} \right]^{-1} \\
 &= \left[ \frac{-1}{n} \sum_{i=1}^n \frac{y_i}{\hat{\theta}^2} \right]^{-2} \frac{1}{n}\sum_{i=1}^n \left(\frac{y_i}{\hat{\theta}} - 1 \right)^2 \\
 &= \left[ \frac{-\hat{\theta}}{\hat{\theta}^2} \right]^{-2} \frac{1}{n}\sum_{i=1}^n \left(\frac{y_i}{\hat{\theta}} - 1 \right)^2 \\
 &= \left[ \frac{-1}{\hat{\theta}} \right]^{-2} \frac{1}{n}\sum_{i=1}^n \left(\frac{y_i}{\hat{\theta}} - 1 \right)^2 \\
 &= \hat{\theta}^2 \frac{1}{n}\sum_{i=1}^n \left(\frac{y_i}{\hat{\theta}} - 1 \right)^2 \\
 &= \hat{\theta}^2 \frac{1}{n\hat{\theta}^2}\sum_{i=1}^n \left(y_i - \hat{\theta} \right)^2 \\
 &= \frac{1}{n}\sum_{i=1}^n \left(y_i - \hat{\theta} \right)^2
\end{align}

From the above, we see that the sandwich estimate of variance for $\hat{\theta}$ is $\tfrac{1}{n}\sum_{i=1}^n (y_i - \hat{\theta})^2$, while the corresponding variance factor in the likelihood pivot is $\tfrac{1}{n}\sum_{i=1}^n (y_i - \theta^{\ast})^2$.  The only difference between the formulas in this example is the substitution of $\theta^{\ast}$ in the pivot quantity with $\hat{\theta}$ in the sandwich estimate of variance.

For the simulation study, I used a negative binomial with mean $3$ and variance $3.9$ as the true data-generating distribution.  I then generated $10,000$ samples of size ranging from $10$ to $100$.  For each sample, I constructed asymptotically correct 95\% confidence intervals based on the sandwich estimate of variance and its post hoc adjustments.  I also constructed a confidence interval by conducting a linear search on the pivot to find the parameter values that caused the pivot to equal the $2.5\%$ and $97.5\%$ quantiles of the standard normal distribution.  The results can be seen in Figure~\ref{poissim}

\begin{figure}
\centering
\includegraphics[width=0.5\textwidth]{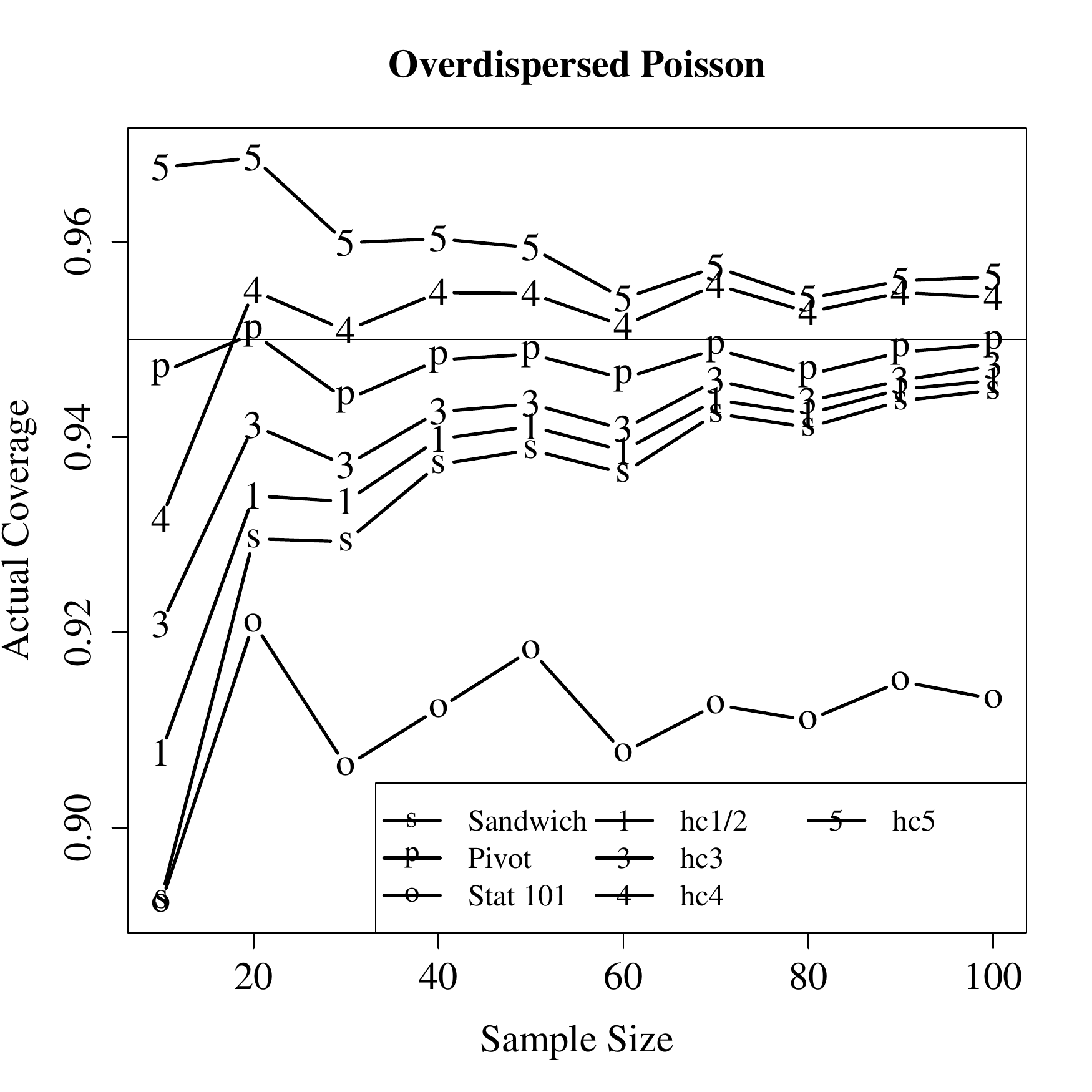}
\caption{Results from simulation study using sandwich-based confidence intervals and confidence intervals generated by the likelihood pivot when applied to an assumed Poisson data-generating model.}
\label{poissim}
\end{figure}

The pivot-based confidence intervals perform well at all sample sizes.  Coverage hovers around $95\%$ throughout all sample sizes tested.  Standard maximum likelihood theory undercovers the true parameter, a reflection of the fact that we are assuming the mean and variance are equal when the variance is actually larger.  The sandwich estimate of variance performs as one would expect.  It undercovers the true mean in all sample sizes tested, but coverage improves as the sample size increases.  The post hoc intervals tend to mirror this as well.  The post hoc intervals associated with \citet{fay2001small} actually tend to overcover the true parameter.  \citet{fay2001small} use an adjusted sandwich and the Student t-distribution quantiles instead of normally distributed quantiles.  This two-pronged approach might be over-correcting the sandwich's undercoverage, resulting in the overcoverage we see in the simulation results.  

In this simple scenario, we see that all confidence interval construction methods that account for model misspecification perform well ($\pm 1$ percentage point), especially as sample sizes increase.  While the likelihood pivot performs the best in this simulation, the improvement is not great.  This is a reflection of the fact that there is only one parameter in our working model, and hence the number of quantities being estimated in any method is small.

\subsubsection{Regression Through the Origin} \label{regthruorig}
Regression models are one of the most fundamental statistical models used by practitioners.  Any technique that improves on methods using the sandwich estimate of variance must work in a regression setting.  In this example, we assume as our working model a two-parameter regression through the origin
\begin{align}
y_i &\sim \theta x_i + \epsilon_i \\
\epsilon_i &\overset{i.i.d.}{\sim} N(0,\sigma^2) \qquad i = 1,\ldots,n
\end{align}
where the parameter of interest is $\theta$ and $\sigma^2$ is a nuisance parameter.  For this working model, the pivot does not depend on $\sigma^2$ so it can be unknown.  We then consider the behavior of the sandwich-based and pivot-based confidence intervals in the case of heteroscedastic errors.  Heteroscedastic errors are a commonly assumed model misspecification, and one of special interest to economists.  Specifically the true model is 
\begin{align}
y_i &\sim \theta x_i + \epsilon_i \\
\epsilon_i &\sim N(0,1 + \mid x_i \mid) \qquad i = 1,\ldots,n
\end{align}
where the errors depend directly on the covariate.  

Under this model, the log-likelihood, its derivative with respect to $\theta$, and the square of the derivative are
\begin{align}
\sum_{i=1}^n l(y_i, \theta) &= -n\log(\sqrt{2\pi}\sigma) + \frac{-1}{2\sigma^2} \sum_{i=1}^n (y_i - \theta x_i)^2 \\
\sum_{i=1}^n l_{\theta}(y_i, \theta) &= \frac{-1}{2\sigma^2} \sum_{i=1}^n -2x_i(y_i - \theta x_i) = \frac{1}{\sigma^2} \sum_{i=1}^n x_i(y_i - \theta x_i) \\
\sum_{i=1}^n l_{\theta}(y_i, \theta)^2 &= \frac{1}{\sigma^4} \sum_{i=1}^n x_i(y_i - \theta x_i)^2
\end{align}
respectively.  Plugging these into line (\ref{mainpivot}), gives us the pivot for this example
\begin{align}
\left\{ \frac{1}{n} \sum_{i=1}^n [l_{\theta}(y_i, \theta^{\ast})]^2 \right\}^{-1/2} \frac{1}{\sqrt{n}} \sum_{i=1}^n l_{\theta}(y_i, \theta^{\ast}) &= \sqrt{n} \frac{\frac{1}{n}\sum_{i=1}^n (x_i y_i - \theta^{\ast} x_i^2)}{\sqrt{\frac{1}{n} \sum_{i=1}^n \left( x_iy_i - \theta^{\ast} x_i^2\right)^2}}.
\end{align}

This simulation was performed similarly to the previous one.  We generated data for samples sizes ranging from $10$ to $100$.  For each sample size, we generated $10,000$ datasets.  For each dataset, we computed $95\%$ confidence intervals using standard MLE theory, the sandwich, the pivot, and the post hoc corrections to the sandwich.  We then computed whether or not the true value of the parameter was covered by the intervals.  The results are displayed graphically in Figure \ref{simplereg}.

\begin{figure}
\centering
\includegraphics[width=0.5\textwidth]{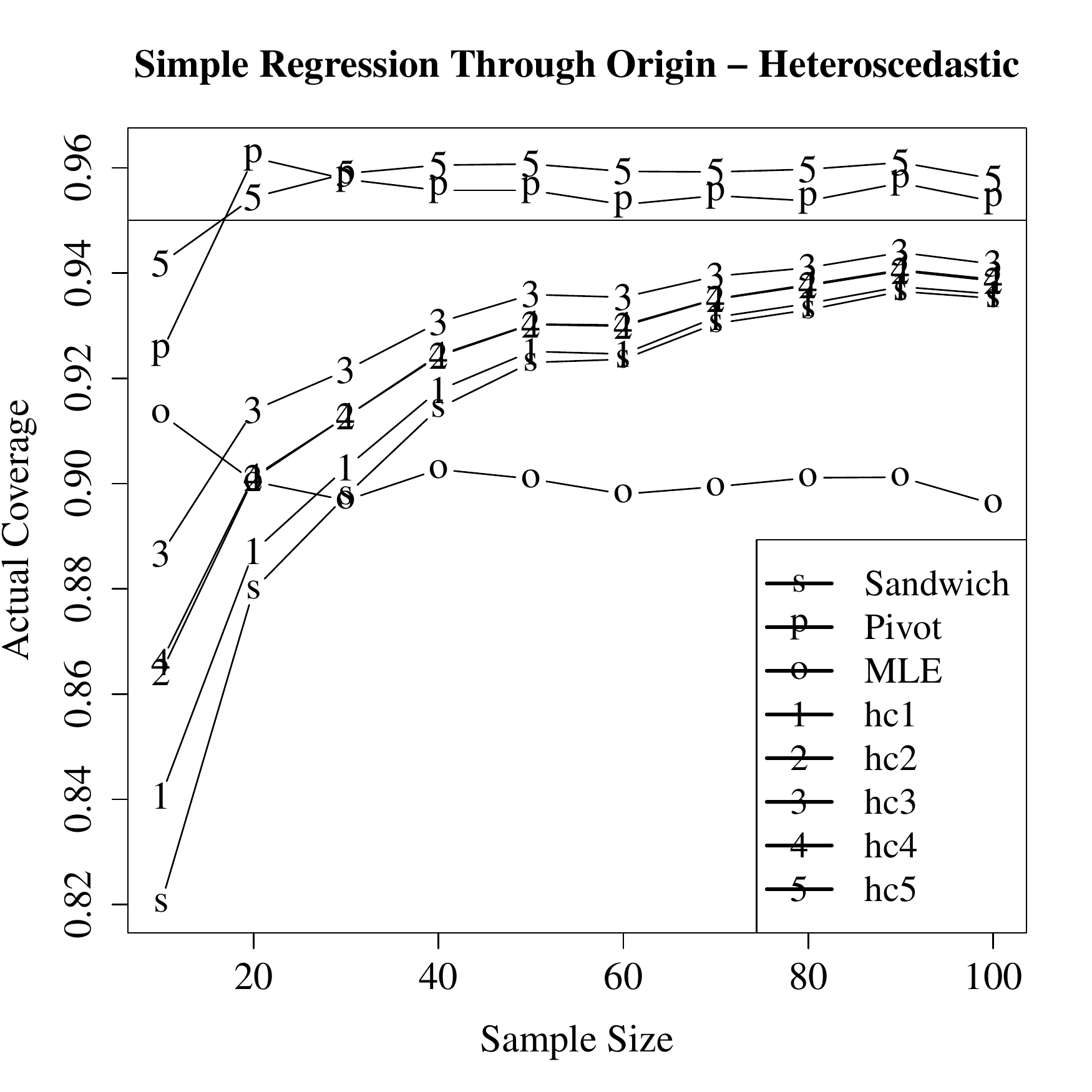}
\caption{Graph showing actual confidence interval coverage of the true slope parameter when using various methods to compute confidence intervals.  For each sample size, 10,000 random datasets were generated, maximum misspecified likelihood estimates were calculated, and confidence intervals were computed.}
\label{simplereg}
\end{figure}

The standard MLE confidence intervals undercover the pseudo-true parameter throughout the range of sample sizes.  Since the assumed model does not account for the variability in errors, we expect confidence intervals to be too small and coverage of the true parameter to be less than nominal.  The sandwich-based confidence intervals yield about $82\%$ coverage at the smallest sample size generated, and increase coverage as sample size increases.  Since the sandwich has this known bias, this is not surprising.  Three of the five post hoc methods, \citep{horn1975estimating}, \citep{hinkley1977jackknifing}, and \citep{efron1982jackknife}, are simply scaled-up versions of the sandwich.  As expected they yield coverage similar to, but slightly better than, the unadjusted sandwich.  The fourth post hoc method uses the sandwich variance but with a slightly wider quantile spread, e.g. the $98\%$ quantile instead of $97.5\%$ quantile.  As such it also covers the true parameter with similar but slightly better coverage than the regular sandwich.  The fifth post hoc method yields similar but slightly higher coverage than the pivot.  Not only does it use an adjusted sandwich variance estimate, but it also uses Student-t quantiles instead of standard normal.  This combination pushes coverage much higher than unaltered sandwich-based coverage.  The pivot-based intervals perform the best out of all options tested and yield coverage of the pseudo-true parameter close to $95\%$ throughout the range of sample sizes tested and versus all competing methods.  

The results of this subsection match what we saw in Subsection \ref{poisdat}.  Confidence intervals based on the likelihood pivot provide coverage closer to nominal than sandwich-based methods when sample sizes are small.  It is no surprise that the post hoc adjustments to the sandwich seem to track one another in Figure \ref{simplereg}, since they generally inflate the sandwich in various ways.  What is of note is how quickly and consistently the likelihood pivot based intervals reach nominal coverage, especially compared to the other intervals.  This has happened in both examples so far.

\section{Multi-Parameter Theory and Examples}\label{multparam}
In this section we consider inference using multi-parameter models.  Confidence regions for sets of parameters are important because they yield reasonable combinations of parameters, something that is not readily apparent from multiple, single-parameter confidence intervals.  We show how to extend the pivot approach to this situation and compare confidence regions of the pivot approach to those based on the sandwich plug-in estimate of variance.  Confidence regions based on the sandwich estimate of variance are values of $\bm{\theta}$ such that
\begin{align} \label{sandmatineq}
n (\hat{\bm{\theta}}_n - \bm{\theta})^T \left[ \mathbf{A}_n(\hat{\bm{\theta}}_n) \mathbf{B}_n(\hat{\bm{\theta}}_n)^{-1} \mathbf{A}_n(\hat{\bm{\theta}}_n) \right] (\hat{\bm{\theta}}_n - \bm{\theta}) &\leq \chi^2_{p, 1-\alpha}
\end{align}
where $p$ is the dimension of the parameter of interest and $\chi^2_{p, 1-\alpha}$ is the $100(1-\alpha) \%$ quantile of a chi-squared distribution with $p$ degrees of freedom.  

Since we are assuming model misspecification, the only distributional result we have is the multivariate central limit theorem.  This gives us an approximate distribution for the multivariate pivot
\begin{align}
\frac{1}{\sqrt{n}} \sum_{i=1}^n l_{\bm{\theta}}(y_i, \bm{\theta}) &\overset{\cdot}{\sim} N(\mathbf{0}_p, \mathbf{B}(\bm{\theta})). \label{approxmultpiv}
\end{align}
We know from standard theory for multivariate distributions that if 
\begin{align}
\mathbf{Z} &\sim N(\mathbf{0}_p, \mathbf{B})
\end{align}
then 
\begin{align}
\mathbf{Z}^T \mathbf{B}^{-1} \mathbf{Z} &\sim \chi^2_p. \label{exactmultpiv}
\end{align}
So combining (\ref{approxmultpiv}) and (\ref{exactmultpiv}) we have that
\begin{align}
\left[ \frac{1}{\sqrt{n}} \sum_{i=1}^n l_{\bm{\theta}}(y_i, \bm{\theta}) \right]^T \mathbf{B}(\bm{\theta})^{-1} \left[ \frac{1}{\sqrt{n}} \sum_{i=1}^n l_{\bm{\theta}}(y_i, \bm{\theta}) \right] &\overset{\cdot}{\sim} \chi^2_p.
\end{align}
We can now generate a confidence region through the pivot by finding solutions, $\bm{\theta}$ to the inequality
\begin{align} \label{pivotmatineqprac}
n \left[ \frac{1}{\sqrt{n}} \sum_{i=1}^n l_{\bm{\theta}}(y_i, \bm{\theta}) \right]^T \mathbf{B}(\bm{\theta})^{-1} \left[ \frac{1}{\sqrt{n}} \sum_{i=1}^n l_{\bm{\theta}}(y_i, \bm{\theta}) \right] &\leq \chi^2_{p, 1-\alpha}
\end{align}
We wish to use an expression that more directly comparable to line \ref{sandmatineq}.  We will use a multivariate equivalent to Jennrich's Lemma 6 (Mean Value Theorem for functions of random variables) \citep{jennrich1969asymptotic}, to equate $\tfrac{1}{\sqrt{n}} \sum_{i=1}^n l_{\bm{\theta}}(y_i, \bm{\theta}) $ and $\overline{\mathbf{A}_n(\bm{\theta})} (\hat{\bm{\theta}}_n - \bm{\theta})$, up to a minus sign.  For our analytical work we will now use the equivalent statement
\begin{align} \label{pivotmatineq}
n (\hat{\bm{\theta}}_n - \bm{\theta})^T \overline{\mathbf{A}_n(\bm{\theta})} \mathbf{B}_n(\bm{\theta})^{-1} \overline{\mathbf{A}_n(\bm{\theta})} (\hat{\bm{\theta}}_n - \bm{\theta}) &\leq \chi^2_{p, 1-\alpha}.
\end{align}
The first difference between the multivariate pivot quantity and the similar sandwich-based quantity is in using $\bm{\theta}$ instead of $\hat{\bm{\theta}}_n$ in $\mathbf{B}_n(\cdot)^{-1}$.  This is equivalent to eliminating the plug-in assumption from the univariate case.  The second is the introduction of $\overline{\mathbf{A}_n(\bm{\theta})}$.  A few words need to be said about this second substitution.  Primarily, there is no single point $\bm{\theta}$ that satisfies Jennrich's Lemma 6 in the multivariate setting.  The reason is that there is no mean value theorem for vector-valued functions.  $\overline{\mathbf{A}_n(\bm{\theta})}$ is really a convex combination of matrices as described in \citet{furi1991mean}.  That is, there are $\bm{\theta}_1, \ldots, \bm{\theta}_k$ and $\omega_1, \ldots, \omega_k \geq 0$ with $\sum_{i=1}^k \omega_i = 1$ such that $\overline{\mathbf{A}_n(\bm{\theta})} = \sum_{i=1}^k \omega_i \mathbf{A}_n(\bm{\theta}_i)$.  Each $\bm{\theta}_i$ is a point in $p$-space where each coordinate is taken from either $\hat{\bm{\theta}}_n$ or $\bm{\theta}$, thus, $k \leq 2^p$.  When we talk about points that satisfy the pivotal quantity, we are referring to solutions $\bm{\theta}$ and our $k$ points are derived from $\hat{\bm{\theta}}$ and $\bm{\theta}$.

\subsection{Multi-Parameter Asymptotics}
\citet{lehmann1999elements} defines multivariate efficiency in terms of the difference of the asymptotic covariance matrices of two estimators.  Since we are not computing an explicit estimator, we do not have a covariance matrix, and we must instead consider what it means to be asymptotically equally efficient.  In the univariate case, we considered the finite boundary of two one-sided confidence intervals, and equivalently, the length of two-sided intervals.  The multivariate equivalent to univariate length is volume.  Thus for a given confidence level, two confidence-region-generating procedures are asymptotically equally efficient if their volumes are asymptotically equivalent.  In this subsection we will make the previous statement precise, explain why it is reasonable, and state a result in the multiple parameter case that is similar to the result for the one parameter case.

Let us define for a given $\alpha$ the sandwich equality and pivot equality respectively as the or-equal-to option in Inequalities (\ref{sandmatineq}) and (\ref{pivotmatineqprac}).
\begin{mythm} \label{multiasythm}
For any solution, $\breve{\bm{\theta}}_{n,PV}$, to the pivot equality and a corresponding solution, $\breve{\bm{\theta}}_{n,SW}$, to the sandwich equality in the same direction from $\hat{\bm{\theta}}_n$ as $\breve{\bm{\theta}}_{n,PV}$, $\sqrt{n} \| \breve{\bm{\theta}}_{n,PV} - \breve{\bm{\theta}}_{n,SW} \| \rightarrow_{a.s.} 0$ as $n\rightarrow \infty$.
\end{mythm}
\begin{proof}
In the Appendix.
\end{proof}

In one sense, this result is the multivariate equivalent of the the result in line (\ref{uniasyequ}).  In the univariate case, the boundary of a one-sided confidence interval is a point.  Subsection \ref{oneparamasy} proves that the boundary points for confidence intervals generated by the sandwich estimate of variance and the pivot converge in norm, even when multiplied by $\sqrt{n}$.  Here we state that the boundary points of the multivariate confidence region converge pointwise, when pointwise is taken to mean a fixed direction from the sequence of maximum misspecified likelihood estimators  $\{\hat{\bm{\theta}}_n\}$.  We also state that this convergence persists under multiplication by $\sqrt{n}$.  

In another sense, this result is one way to formalize the statement ``the volumes of the confidence regions are asymptotically equivalent''.  While not computing volumes explicitly, we are demonstrating that the boundaries of the respective confidence regions are converging to each other.  If the boundaries of the regions are converging to each other, then the respective volumes must also be converging to each other, which is what it means for the volumes to be asymptotically equivalent.  If the volumes are asymptotically equivalent at equal confidence levels, then the confidence region procedures are asymptotically equally efficient.  

\subsection{Multiparameter Simulation Example}
Having compared the sandwich-based pivot and the likelihood pivot asymptotically, we now consider finite sample comparisons.  In this example we use as our working model a simple linear regression
\begin{align}
y_i &= \theta_0 + \theta_1 x_i + \epsilon_i \label{multSLR}\\
\epsilon_i &\overset{i.i.d.}{\sim} N(0, \sigma^2) \label{multSLRerr}
\end{align}
where the parameter of interest is $\bm{\theta} = (\theta_0, \theta_1)^T$ and $\sigma^2$ is a nuisance parameter.  As in the regression through the origin example, we can assume $\sigma^2$ known, but since it cancels in the pivot we could also assume it is an unknown without loss of generality.  

We evaluate pivot confidence regions based on the model described by (\ref{multSLR}) and (\ref{multSLRerr}) when the true data-generating mechanism is 
\begin{align}
y_i &= \theta_0 + \theta_1 x_i + \epsilon_i \\
\epsilon_i &\overset{indep}{\sim} N(0, 1+|x_i|)
\end{align}
with heteroscedastic errors depending on the covariate.  Under this model the relevant term of the pseudo-log-likelihood and necessary partial derivatives are given below:
\begin{align}
\sum_{i=1}^n l(y_i, \bm{\theta}) &= -n\log(\sqrt{2\pi}\sigma) + \frac{-1}{2\sigma^2} \sum_{i=1}^n (y_i - \theta_0 - \theta_1 x_i)^2 \\
\sum_{i=1}^n l_{\theta_0}(y_i, \bm{\theta}) &= \frac{1}{2\sigma^2} \sum_{i=1}^n (y_i - \theta_0 - \theta_1 x_i) \\
\sum_{i=1}^n l_{\theta_1}(y_i, \bm{\theta}) &= \frac{1}{2\sigma^2} \sum_{i=1}^n x_i(y_i - \theta_0 - \theta_1 x_i) .
\end{align}
We can write the factors in the pivot inequality (\ref{pivotmatineqprac}) as
\begin{align}
l_{\bm{\theta}} &=
\begin{pmatrix}
\frac{1}{2n\sigma^2} \sum_{i=1}^n (y_i - \theta_0 - \theta_1 x_i)  \\
\frac{1}{2n\sigma^2} \sum_{i=1}^n x_i(y_i - \theta_0 - \theta_1 x_i) 
\end{pmatrix} \\
\mathbf{B}_n(y,\bm{\theta}) &=
\begin{pmatrix}
\frac{1}{4n\sigma^4} \sum_{i=1}^n (y_i - \theta_0 - \theta_1 x_i)^2  & \frac{1}{4n\sigma^4} \sum_{i=1}^n x_i(y_i - \theta_0 - \theta_1 x_i)^2 \\
\frac{1}{4n\sigma^4} \sum_{i=1}^n x_i(y_i - \theta_0 - \theta_1 x_i)^2  & \frac{1}{4n\sigma^4} \sum_{i=1}^n x_i^2(y_i - \theta_0 - \theta_1 x_i)^2
\end{pmatrix}
\end{align}
and then the pivot is exactly the one as in inequality (\ref{pivotmatineqprac}).

A simulation example was performed similarly to the one-parameter examples.  We generated data for samples sizes ranging from $10$ to $100$.  For each sample size, we generated $10,000$ datasets.  For each dataset, we computed whether or not the true value of the parameter was covered by the specified confidence regions.  The results are shown in Figure \ref{simpleregmult}.  We only show results for the sandwich, the pivot, and the hc1, hc2, and hc3 sandwich adjustments.  The hc4 and hc5 methods were designed for scalar, linear combinations of parameters and are not adaptable to a multivariate setting.  

\begin{figure}
\centering
\includegraphics[width=0.5\textwidth]{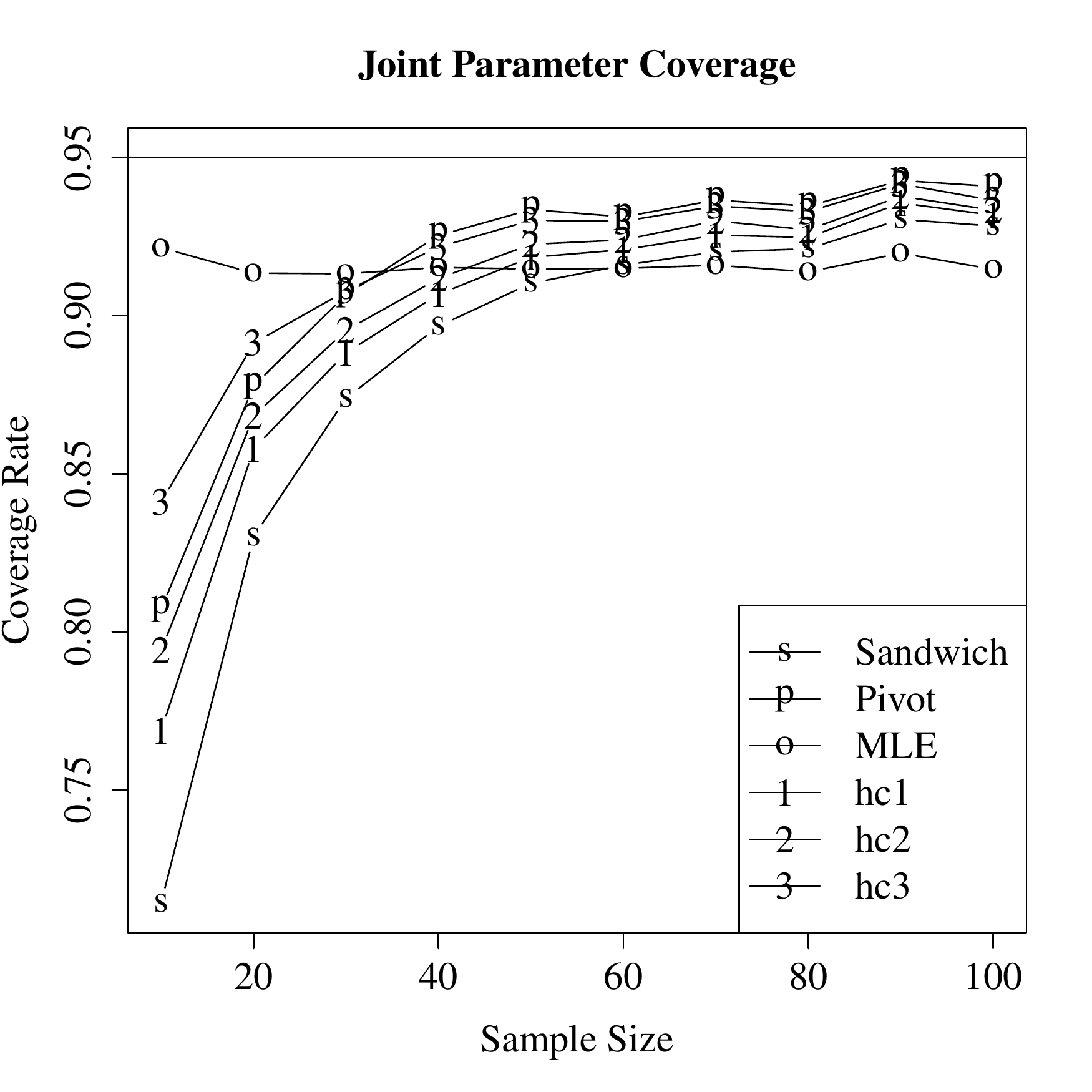}
\caption{Graph showing confidence interval coverage of the true parameter vector using various methods.}
\label{simpleregmult}
\end{figure}

The results in this multi-parameter simulation mirror the results from the one-parameter simulations.  The confidence regions from standard maximum likelihood theory undercover the parameter vector consistently throughout all sample sizes tested.  The sandwich-based confidence regions undercover the pseudo-true parameter, with coverage improving as sample size increases.  The post hoc methods all provide slightly better coverage than the sandwich, as they all increase the sandwich variance estimate.  The pivot-based confidence regions provide comparable coverage to the post hoc methods, providing the best coverage in this simulation in all but the very smallest sample sizes. 

We note that all methods which account for misspecification show poor coverage at the smallest sample sizes.  This coverage is much worse than what was seen in the one-parameter examples.  We think that this is related to the fact that as the dimension of the parameter increases, the number of quantities estimated in the covariance matrix increases with the square of said dimension.  As will be seen in the next section, coverage can be quite poor when the parameter vector is even modest in dimension.  

\section{Data Example}\label{realworld}
In this section we look at an example using a real dataset as our population.  A more realistic comparison of confidence interval methods can be made by using real data as opposed to simulated data.  Real world model misspecifications are not just limited to heteroscedastic errors either, so this allows us to compare the methods in the presence of more general model misspecification.

The data comes from the National Health and Nutrition Examination Survey (NHANES) dataset, 2011-2012.  We consider a scenario in which a researcher is interested in modeling body-mass index (BMI) as a function of several explanatory variables using a linear regression model.  The explanatory variables are taken from variables related to food consumption and activity level, factors that might reasonably affect an individual's BMI.  

We use as our population the $7,804$ complete cases in the dataset.  The pseudo-true parameters of the model will be the ordinary least-squares estimates applied to these $7,804$ cases.  We ran a linear regression of BMI on all the selected explanatory variables to get our pseudo-true parameter values.  The model we used was a standard regression model
\begin{align}
\text{BMI} &= \theta_0 + \theta_1 \times \text{kcal} + \theta_2 \times \text{sugar} + \theta_3 \times \text{fat} + \theta_4 \times \text{inact} + \theta_5 \times \text{gender} + \epsilon \\
\epsilon &\sim N(0, \sigma^2)
\end{align}
with ``kcal'' being the kilocalorie consumption, ``inact'' being a numerical measurement of inactivity, and $\epsilon_i$ representing additional across-subject variability.  Next we sampled small datasets from the population to mimic performing multiple small-scale studies.  We then ran regressions using those samples, and computed confidence region coverage for all covariate coefficients together.  As in the simple linear regression model, we compared the coverage rates of standard maximum likelihood, the sandwich, the pivot, and the hc1, hc2, and hc3 adjustments.  The pivot used was identical in structure to the one in the multivariate simulation example, the only difference being the number of covariates was greater in this example.  The results can be see in Figure \ref{realworldmult}.  

\begin{figure}
\centering
\includegraphics[width=0.5\textwidth]{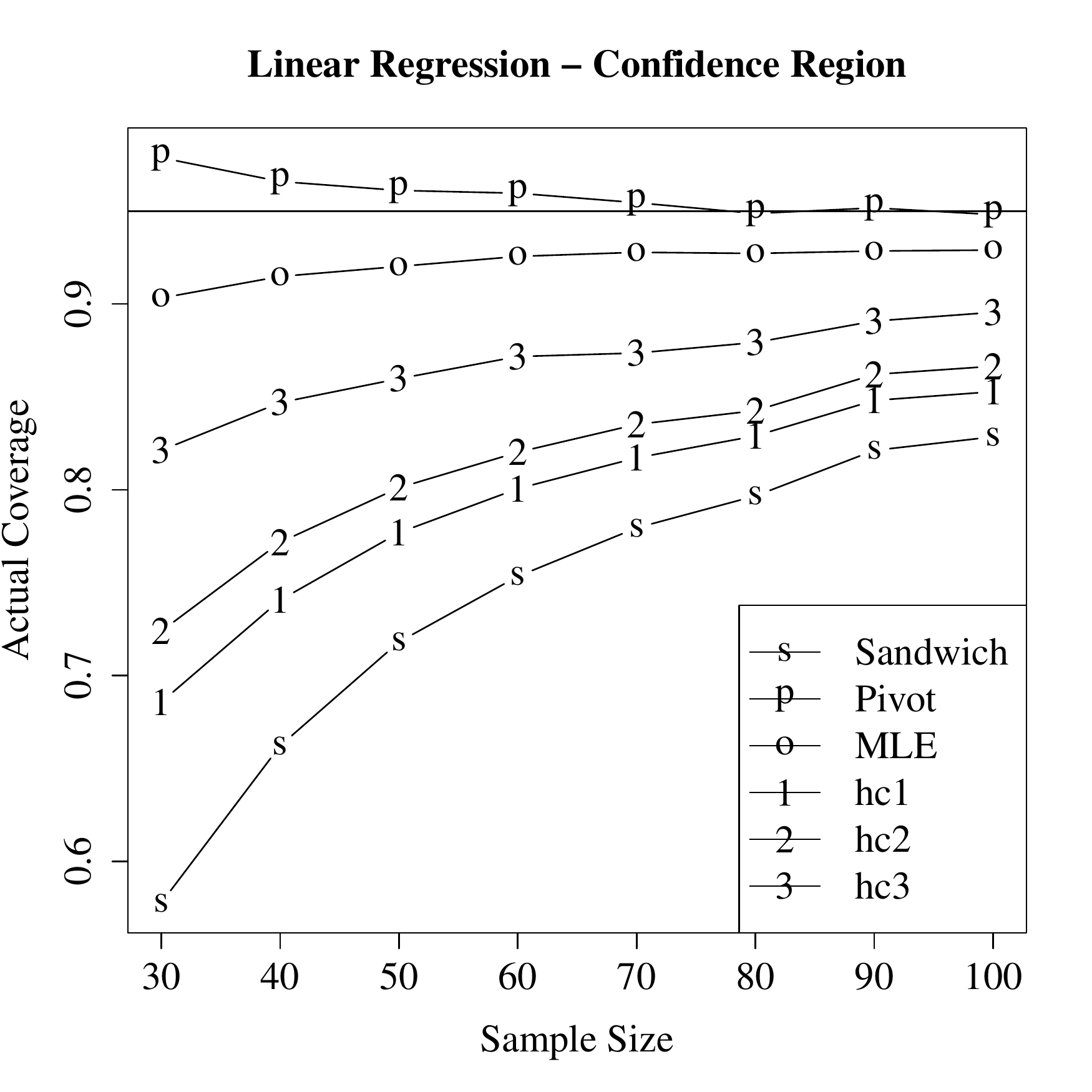}
\caption{Graph showing confidence interval coverage of the pseudo-true parameter vector value in the NHANES data using various methods.  The horizontal line is the 95\% line.}
\label{realworldmult}
\end{figure}

The standard maximum likelihood method undercovers the pseudo-true vector parameter, showing just under $93\%$ coverage at the higher sample sizes.  The sandwich does particularly poorly at first, but improves as sample size increases.  The post hoc methods known as hc1, hc2, and hc3 do slightly better than the sandwich, again because they start with the sandwich and increase the variance estimates by varying amounts.  The pivot starts out with slightly high coverage, but quickly drops to $95\%$, and is the only method that achieves $95\%$ coverage within this simulation.  

The results of this example are consistent with the results in previous simulations.  What this example highlights is that the small sample, model misspecification scenarios under which the likelihood pivot based confidence intervals perform better than sandwich based intervals need not be limited to the examples examined in previous sections.  It also highlights the theme running throughout all of these examples:  replacing $\theta^{\ast}$ with $\hat{\theta}$ can seriously lower confidence interval and confidence region coverage.  This can cause asymptotically correct inference to differ significantly from nominally correct inference.  

\section{Discussion}\label{discuss}
We have presented a pivot-based method of inference that outperforms the sandwich and its variants in small sample sizes.  We further showed that this method incurs no efficiency loss at large sample sizes, when compared to sandwich-based inference.  Several simulation examples were presented which strengthen our argument that making fewer assumptions leads to better inference.  

A natural extension of this work would be to use pivots in a Bayesian context.  The seminal paper on Bayesian work with pivots is \citet{monahan1992proper}.  The main result of that paper is that likelihoods of pivotal quantities are, in general, not true likelihoods.  Consequently, proper Bayesian analysis using pivots is impossible in most cases.  However, we believe that some kind of pseudo-Bayesian analysis might be feasible.  Our reasoning is simple:  for a given set of data, a pivotal quantity should contain some information about the parameter of interest.  That information is the basis of the frequentist analysis performed in this paper.  A well-designed pseudo-Bayesian analysis might be able to take advantage of that information as well.  

\appendix
\section{Univariate Results}

We will show that $B(\theta^{\ast})$ and $A(\theta^{\ast})$ are positive on $\Theta$.  We also need to show that $A_n(\theta^{\ast})$ is positive on this same subset, but due to the uniform central limit theorems, it suffices to show that $A(\theta^{\ast})$ is positive on the required subset.  Hence, for large enough $n$, $A_n(\theta^{\ast})$ is also positive on that subset.

\begin{mylemma}
Under the assumptions of \citep{white1982maximum}, $B(\theta^{\ast}) >0$.  
\end{mylemma}
\begin{proof}
By definition,
\begin{align}
B(\theta) &= E[l_{\theta}(y; \theta)^2]
\end{align}
where the expectation is taken with respect to the true density.  Since we're taking the expectation of the square of a function, we have that
\begin{align}\label{Bnonneg}
B(\theta) &\geq 0
\end{align}
for all possible $\theta \in \Theta$.  

Assumption A6 in \citep{white1982maximum} states that $B(\theta^{\ast})$ is nonsingular.  For the one-dimensional case, that means that $B(\theta^{\ast}) \neq 0$.  Combined with (\ref{Bnonneg}), we have $B(\theta^{\ast}) > 0$.  
\end{proof}

\begin{mylemma}
Under the assumptions of \citep{white1982maximum}, $A(\theta^{\ast}) >0$.
\end{mylemma}
\begin{proof}
By definition,
\begin{align}
A(\theta) &= -E[l_{\theta \theta}(y; \theta)].
\end{align}
The case when this is positive is the same as when the expectation of $l_{\theta \theta}$ is negative.  Suppose the contrary, namely that $\theta^{\ast}$ is such that $A(\theta^{\ast}) < 0$, or the expectation of $l_{\theta \theta}$ is positive.  Since $A_n(\hat{\theta}) \rightarrow_{a.s.} A(\theta^{\ast})$, for $n$ large enough, we have that $A_n(\hat{\theta}) < 0$.  In other words, the second derivative of the log-likelihood is positive, or in other words that the derivative of the score equation evaluated at $\theta^{\ast}$ is positive.  But this would mean that $\hat{\theta}$ is a minimum of the likelihood equation, and not a maximum.  Hence, we restrict ourselves to the cases where $A(\theta^{\ast})\geq 0$.  

Assumption A6 in \citep{white1982maximum} further states that $\theta^{\ast}$ is a ``regular'' point of $A(\theta)$.  \citet{white1982maximum} then defines a regular point as one in which $A(\theta)$ has constant rank in an open neighborhood of that point.  For the one-dimensional case, we have that if $A(\theta)=0$, then $A(\theta)$ has rank $0$.  If $A(\theta) \neq 0$, then $A(\theta)$ has rank $1$.  So, if $A(\theta^{\ast}) \geq 0$, then for the constant rank assumption to hold we must have $A(\theta^{\ast}) > 0$ or $A(\theta^{\ast})$ is identically $0$ in a neighborhood of $\theta_{\ast}$.  But if the latter condition holds, then $E[l_{\theta}]$ is constant in a neighborhood of $\theta^{\ast}$, and it no longer has a unique minimum at $\theta^{\ast}$.  Hence, $A(\theta^{\ast}) > 0$.  
\end{proof}

We now turn our attention to $\theta_{1,n}$ and $\theta_{2,n}$.  We will show that $\theta_{1,n}\rightarrow_{a.s.} \theta^{\ast}$ and $\theta_{2,n}\rightarrow_{a.s.} \theta^{\ast}$.  

\begin{mylemma}
Under the assumptions of \citep{white1982maximum}, $\theta_{2,n}\rightarrow_{a.s.} \theta^{\ast}$.  
\end{mylemma}
\begin{proof}
We can solve directly for $\theta_{2,n}$ in terms of other quantities:
\begin{align}
\theta_{2,n} &= \hat{\theta}_n - \frac{A_n(\hat{\theta}_n)^{-1} B_n(\hat{\theta}_n)^{1/2} z_{1-\alpha}}{\sqrt{n}}.
\end{align}
We know from \citep{white1982maximum} that $\hat{\theta}_n\rightarrow_{a.s.}\theta^{\ast}$.  By Lemma 4 of \citep{amemiya1973regression}, this means that $B_n(\hat{\theta}_n)\rightarrow_{a.s.}B(\theta^{\ast})$ and $A_n(\hat{\theta}_n)\rightarrow_{a.s.}A(\theta^{\ast})$.  By Mann-Wald this implies that $B_n(\hat{\theta}_n)^{1/2}\rightarrow_{a.s.}B(\theta^{\ast})^{1/2}$ and $A_n(\hat{\theta}_n)^{-1}\rightarrow_{a.s.}A(\theta^{\ast})^{-1}$.  $z_{1-\alpha}$ is a constant.  And $\sqrt{n}$ grows arbitrarily large with increasing $n$.  Thus, we have $\theta_{2,n} \rightarrow_{a.s.} \theta^{\ast} - 0 = \theta^{\ast}$.  
\end{proof}

\begin{mylemma}
For all $\epsilon > 0$ there exists $N$, $A_{max}$, $A_{min}$, $B_{max}$, and $B_{min}$ such that for all $n > N$, 
\begin{align}
A_{min} - \epsilon &\leq A_n(\bar{\theta}_n)^{-1} \leq A_{max} + \epsilon \\
B_{min} -\epsilon &\leq B_n(\theta_{1,n})^{1/2} \leq B_{max} + \epsilon.
\end{align}
\end{mylemma}
\begin{proof}
At this point, we remind the reader of some of the standard regularity assumptions generally made in this subject.  Specifically, we will assume that the $\theta$-space, $\Theta$, is a compact subset of Euclidean space.  We will also rely on Assumption A5 in \citep{white1982maximum} to ensure the continuity of $B(\theta)$ and $A(\theta)$.  Furthermore, we will be able to apply a uniform law of large numbers to $B_n(\theta)$ and $A_n(\theta)$.  A full list of standard regularity assumptions can be found in \citep{white1982maximum}.

Additionally, we need to be working in a convex set.  Technically, $\Theta$ is only a compact set in Euclidean space.  But, being closed and bounded, we can enclose $\Theta$ within a possibly larger, convex, still compact set $\Theta_{convex}$.  We will work within $\Theta_{convex}$ to set our bounds.  We also need to assume that $B_n(\theta)^{1/2}$, $B(\theta)^{1/2}$, $A_n(\theta)^{-1}$ and $A(\theta)^{-1}$ are continuous functions on $\Theta_{convex}$.  

Now, since $B_n(\theta)^{1/2}$, $B(\theta)^{1/2}$, $A_n(\theta)^{-1}$ and $A(\theta)^{-1}$ are continuous functions on $\theta \in \Theta_{convex}$ and $\Theta_{convex}$ is compact, those functions each map to compact subsets of $\mathbb{R}$, thus they attain their maximum and minimum on the image sets.  That is, there are values $\theta_{min,B^{1/2}}$ and $\theta_{max,B^{1/2}}$ such that $0 < B(\theta_{min,B^{1/2}})^{1/2} \leq B(\theta)^{1/2} \leq B(\theta_{max,B^{1/2}})^{1/2} < \infty, \, \forall \theta \in \Theta_{convex}$.  We can define $\theta_{min,A^{-1}}$, $\theta_{max,A^{-1}}$, $\theta_{min,A_n^{-1}}$, and $\theta_{max,A_n^{-1}}$, $\theta_{min,B_n^{1/2}}$, and $\theta_{max,B_n^{1/2}}$ similarly.

We will proceed to show the details for the bound on $A_n(\bar{\theta}_n)^{-1}$ knowing that the argument is similar for $B_n(\theta_{1,n})^{1/2}$.  From the convexity of $\Theta_{convex}$, we know that $\bar{\theta}_n = (1-b)\hat{\theta}_n + b \theta_{1,n} \in \Theta_{convex}$ for some $b \in [0,1]$, and hence $A_n((1-b)\hat{\theta}_n + b \theta_{1,n})^{-1} \leq A_n(\theta_{max,A_n^{-1}})^{-1}$.  Since we can apply the uniform law of large numbers to $A_n(\theta)^{-1}$, we have that 
\begin{align}
\sup_{\theta \in \Theta_{convex}} | A_n(\theta)^{-1} - A(\theta)^{-1} | &\rightarrow_{a.s.} 0.
\end{align}
What this means is that with probability one, $\forall \epsilon > 0$, $\exists N_{\epsilon}$ such that $\forall n \geq N_{\epsilon}$, $\forall \theta \in \Theta_{convex}$, $|A_n(\theta) - A(\theta)| < \epsilon$.  So, if we fix an $\epsilon>0$, there is some $N_{\epsilon}$ such that for all $n\geq N_{\epsilon}$ we have
\begin{align}
A_n((1-b)\hat{\theta}_n + b \theta_{1,n})^{-1} &\leq A_n(\theta_{max,A_n^{-1}})^{-1} \label{defnofmax}\\
 &\leq A(\theta_{max,A_n^{-1}})^{-1} + \epsilon \label{uniflln}\\
 &\leq A(\theta_{max,A^{-1}})^{-1} + \epsilon. \label{defnAmax}
\end{align}
Line (\ref{defnofmax}) comes from the definition of $\theta_{max,A_n^{-1}}$.  Line (\ref{uniflln}) is from the uniform law of large numbers.  And line (\ref{defnAmax}) is from the definition of $\theta_{max,A^{-1}}$.  

Now, we can say that for a fixed $\epsilon > 0$, there exists an $N_{\epsilon, B^{1/2}}$ for $B_n(\theta)^{1/2}$ and an $N_{\epsilon, A_n^{-1}}$ for $A(\theta)^{-1}$ such that for $n \geq max(N_{\epsilon, B^{1/2}}, N_{\epsilon, A_n^{-1}})$
\begin{align}
A_n((1-b)\hat{\theta}_n + b \theta_{1,n})^{-1} &\leq A(\theta_{max,A^{-1}})^{-1} + \epsilon \\
B_n(\theta_{1,n})^{1/2} &\leq B(\theta_{max,B^{1/2}})^{1/2} + \epsilon.
\end{align}
By the same style of argument for $n$ sufficiently large, we have additionally that
\begin{align}
A_n((1-b)\hat{\theta}_n + b \theta_{1,n})^{-1} &\geq A(\theta_{min,A^{-1}})^{-1} - \epsilon \\
B_n(\theta_{1,n})^{1/2} &\geq B(\theta_{min,B^{1/2}})^{1/2} - \epsilon.
\end{align}
If we let $A_{max} = A(\theta_{max,A^{-1}})^{-1}$, $B_{max} = B(\theta_{max,B^{1/2}})^{1/2}$, $A_{min} = A(\theta_{min,A^{-1}})^{-1}$, $B_{min} = B(\theta_{min,B^{1/2}})^{1/2}$ then the result is proved.
\end{proof}

\begin{mylemma}
$\theta_{1,n} \rightarrow_{a.s.} \theta^{\ast}$
\end{mylemma}
\begin{proof}
$\theta_{1,n}$ solves the following equation
\begin{align}
\sqrt{n} B_n(\theta_{1,n})^{-1/2} \frac{1}{n}\sum_{i=1}^n l_{\theta}(y_i; \theta_{1,n}) &= z_{1-\alpha}.
\end{align}
Going back to Jennrich's Lemma, we see that the result still holds replacing $\theta^{\ast}$ with $\theta_{1,n}$.  Thus, we can investigate
\begin{align} \label{pivottheta1n}
\sqrt{n} B_n(\theta_{1,n})^{-1/2} A_n(\bar{\theta}) (\hat{\theta} - \theta_{1,n}) &= z_{1-\alpha}
\end{align}
instead.

Since $\bar{\theta}$ lies on the line segment containing $\hat{\theta}$ and $\theta_{1,n}$, we can write $\bar{\theta}=(1-b)\hat{\theta} + b \theta_{1,n}$ for some $b \in [0,1]$.  This means our equation can be re-written as
\begin{align}\label{solveforthetaone}
\sqrt{n} B_n(\theta_{1,n})^{-1/2} A_n((1-b)\hat{\theta} + b \theta_{1,n}) (\hat{\theta} - \theta_{1,n}) &= z_{1-\alpha}, \qquad b \in [0,1].
\end{align}
Clearly, we cannot solve this equation directly for $\theta_{1,n}$ in the general case.  Instead we perform a linear search to find the value of $\theta_{1,n}$ that satisfies Equation~\eqref{solveforthetaone}.  

We do some minor algebra to the previous equation 
\begin{align}
\hat{\theta}_n - \theta_{1,n} &= \frac{z_{1-\alpha}}{\sqrt{n}} B_n(\theta_{1,n})^{1/2} A_n((1-b)\hat{\theta}_n + b \theta_{1,n})^{-1}
\end{align}

This means that for $n$ sufficiently large, we have 
\begin{align}
-\infty &< \frac{z_{1-\alpha}}{\sqrt{n}} (B(\theta_{min,B^{1/2}})^{1/2}-\epsilon) (A(\theta_{min,A^{-1}})^{-1} - \epsilon) \\
 &\leq \hat{\theta}_n - \theta_{1,n} \leq \frac{z_{1-\alpha}}{\sqrt{n}} (B(\theta_{max,B^{1/2}})^{1/2}+\epsilon) (A(\theta_{max,A^{-1}})^{-1} + \epsilon) < \infty.
\end{align}
But the products $z_{1-\alpha} (B(\theta_{min,B^{1/2}})^{1/2}-\epsilon) (A(\theta_{min,A^{-1}})^{-1} - \epsilon)$, and $z_{1-\alpha} (B(\theta_{max,B^{1/2}})^{1/2}+\epsilon) (A(\theta_{max,A^{-1}})^{-1} + \epsilon)$ are constants, so the middle term converges to zero as $n\rightarrow \infty$ by the Squeeze Theorem.  And since $\hat{\theta}_n \rightarrow_{a.s.} \theta^{\ast}$, we have that $\theta_{1,n} \rightarrow_{a.s.} \theta^{\ast}$ as well.
\end{proof}

\section{Multivariate Results}

\begin{mylemma}\label{convexsum}
Suppose $\mathbf{A}_1, \ldots, \mathbf{A}_k, \mathbf{A}$ are matrices such that $\| \mathbf{A}_i - \mathbf{A} \| < \epsilon$ for $i = 1,\ldots,k$ and for some $\epsilon > 0$.  Let $\omega_1, \ldots,\omega_k$ be weights such that $\omega_i \geq 0$ for $i = 1,\ldots,k$ and $\sum_{i=1}^k \omega_i = 1$.  Then $\| \sum_{i=1}^k \omega_i \mathbf{A}_i - \mathbf{A} \| < \epsilon$.
\end{mylemma}
\begin{proof}
We have
\begin{align}
\| \sum_{i=1}^k \omega_i \mathbf{A}_i - \mathbf{A} \| &= \| \sum_{i=1}^k \omega_i (\mathbf{A}_i - \mathbf{A}) \| \\
 &\leq \sum_{i=1}^k \| \omega_i (\mathbf{A}_i - \mathbf{A}) \| \\
 &= \sum_{i=1}^k \omega_i \| (\mathbf{A}_i - \mathbf{A}) \| \\
 &\leq \sum_{i=1}^k \omega_i \epsilon \\
 &= \epsilon
\end{align}
\end{proof}

\begin{mylemma}\label{abaconv}
For every $\epsilon > 0$ there is a ball of radius $\delta$ around $\bm{\theta}^{\ast}$, called $Ball_{\delta}(\bm{\theta}^{\ast})$, such that if $\bm{\theta}_0, \bm{\theta}_1, \ldots, \bm{\theta}_k \in Ball_{\delta}(\bm{\theta}^{\ast})$ and $\mathbf{A}(\bar{\bm{\theta}}) = \sum_{i=1}^k \omega_i \mathbf{A}(\bm{\theta}_i)$ a convex combination, then $\| \mathbf{A}(\bm{\theta}^{\ast}) \mathbf{B}(\bm{\theta}^{\ast}) \mathbf{A}(\bm{\theta}^{\ast}) - \mathbf{A}(\bar{\bm{\theta}}) \mathbf{B}(\bm{\theta}_0) \mathbf{A}(\bar{\bm{\theta}}) \| < \epsilon$.
\end{mylemma}
\begin{proof}
First, we note the constants $A^{\ast} = \| \mathbf{A}(\bm{\theta}^{\ast}) \|$ and $B^{\ast} = \| \mathbf{B}(\bm{\theta}^{\ast}) \|$.  

Next, we use the continuity of $\mathbf{B}(\bm{\theta})$ with respect to $\bm{\theta}$ to see that for all $\epsilon > 0$, there is a $\delta_B > 0$ such that if $\| \bm{\theta} - \bm{\theta}^{\ast} \| < \delta_B$ then $\| \mathbf{B}(\bm{\theta}) - \mathbf{B}(\bm{\theta}^{\ast}) \| < \epsilon / (3 A^{\ast})$.

We now set $B_{bnd} = \underset{\| \bm{\theta} - \bm{\theta}^{\ast} \| < \delta_B}{\sup} \| \mathbf{B}(\bm{\theta}) \|$ and similarly set $A_{bnd} = \underset{\| \bm{\theta} - \bm{\theta}^{\ast} \| < \delta_B}{\sup} \| \mathbf{A}(\bm{\theta}) \|$.  

This implies that for all $\epsilon > 0$, there is a $\delta_{B'} > 0$ such that if $\| \bm{\theta} - \bm{\theta}^{\ast} \| < \delta_{B'}$ then $\| \mathbf{B}(\bm{\theta}) - \mathbf{B}(\bm{\theta}^{\ast}) \| < \epsilon / (3 A^{\ast} A_{bnd})$

Now use the continuity of $\mathbf{A}(\bm{\theta})$ with respect to $\bm{\theta}$ to see that for all $\epsilon > 0$, there is a $\delta_A > 0$ such that if $\| \bm{\theta} - \bm{\theta}^{\ast} \| < \delta_A$ then $\| \mathbf{A}(\bm{\theta}) - \mathbf{A}(\bm{\theta}^{\ast}) \| < \epsilon / (3*\max(A^{\ast} \cdot B^{\ast}, A_{bnd} \cdot B_{bnd}))$.

Pick $\epsilon > 0$, and let $\delta = \min(\delta_A, \delta_B, \delta_{B'})$.  This means that $B_{bnd} \geq \underset{\| \bm{\theta} - \bm{\theta}^{\ast} \| < \delta}{\sup} \| \mathbf{B}(\bm{\theta}) \|$ and likewise $A_{bnd} \geq \underset{\| \bm{\theta} - \bm{\theta}^{\ast} \| < \delta}{\sup} \| \mathbf{A}(\bm{\theta}) \|$.

Consider $Ball_{\delta}(\bm{\theta}^{\ast}) = \{ \bm{\theta} \mid \| \bm{\theta} - \bm{\theta}^{\ast} \| < \delta \}$.  Now select any $\bm{\theta}_0, \bm{\theta}_1, \ldots, \bm{\theta}_k \in Ball_{\delta}(\bm{\theta}^{\ast})$.  And let $\mathbf{A}(\bar{\bm{\theta}}) = \sum_{i=1}^k \omega_i \mathbf{A}(\bm{\theta}_i)$ a convex combination.  Then we have

\begin{align}
& \| \mathbf{A}(\bm{\theta}^{\ast}) \mathbf{B}(\bm{\theta}^{\ast}) \mathbf{A}(\bm{\theta}^{\ast}) - \mathbf{A}(\bar{\bm{\theta}}) \mathbf{B}(\bm{\theta}_0) \mathbf{A}(\bar{\bm{\theta}}) \| \\
&= \| \mathbf{A}(\bm{\theta}^{\ast}) \mathbf{B}(\bm{\theta}^{\ast}) \mathbf{A}(\bm{\theta}^{\ast}) - \mathbf{A}(\bm{\theta}^{\ast}) \mathbf{B}(\bm{\theta}^{\ast}) \mathbf{A}(\bar{\bm{\theta}}) + \mathbf{A}(\bm{\theta}^{\ast}) \mathbf{B}(\bm{\theta}^{\ast}) \mathbf{A}(\bar{\bm{\theta}}) - \mathbf{A}(\bar{\bm{\theta}}) \mathbf{B}(\bm{\theta}_0) \mathbf{A}(\bar{\bm{\theta}}) \| \\
&\leq \| \mathbf{A}(\bm{\theta}^{\ast}) \mathbf{B}(\bm{\theta}^{\ast}) \mathbf{A}(\bm{\theta}^{\ast}) - \mathbf{A}(\bm{\theta}^{\ast}) \mathbf{B}(\bm{\theta}^{\ast}) \mathbf{A}(\bar{\bm{\theta}}) \| + \| \mathbf{A}(\bm{\theta}^{\ast}) \mathbf{B}(\bm{\theta}^{\ast}) \mathbf{A}(\bar{\bm{\theta}}) - \mathbf{A}(\bar{\bm{\theta}}) \mathbf{B}(\bm{\theta}_0) \mathbf{A}(\bar{\bm{\theta}}) \| \\
&\leq \| \mathbf{A}(\bm{\theta}^{\ast}) \| \| \mathbf{B}(\bm{\theta}^{\ast}) \| \| \mathbf{A}(\bm{\theta}^{\ast}) - \mathbf{A}(\bar{\bm{\theta}}) \| + \| \mathbf{A}(\bm{\theta}^{\ast}) \mathbf{B}(\bm{\theta}^{\ast}) - \mathbf{A}(\bar{\bm{\theta}}) \mathbf{B}(\bm{\theta}_0) \| \| \mathbf{A}(\bar{\bm{\theta}}) \| \\
 &\leq A^{\ast} B^{\ast} \frac{\epsilon}{3*\max(A^{\ast} B^{\ast}, A_{bnd} B_{bnd})} \\
&\; + \| \mathbf{A}(\bm{\theta}^{\ast}) \mathbf{B}(\bm{\theta}^{\ast}) - \mathbf{A}(\bm{\theta}^{\ast}) \mathbf{B}(\bm{\theta}_0) + \mathbf{A}(\bm{\theta}^{\ast}) \mathbf{B}(\bm{\theta}_0) - \mathbf{A}(\bar{\bm{\theta}}) \mathbf{B}(\bm{\theta}_0) \| \| \mathbf{A}(\bar{\bm{\theta}}) \| \\
 &\leq \frac{\epsilon}{3} + \| \mathbf{A}(\bm{\theta}^{\ast}) \mathbf{B}(\bm{\theta}^{\ast}) - \mathbf{A}(\bm{\theta}^{\ast}) \mathbf{B}(\bm{\theta}_0) \| \| \mathbf{A}(\bar{\bm{\theta}}) \| + \| \mathbf{A}(\bm{\theta}^{\ast}) \mathbf{B}(\bm{\theta}_0) - \mathbf{A}(\bar{\bm{\theta}}) \mathbf{B}(\bm{\theta}_0) \| \| \mathbf{A}(\bar{\bm{\theta}}) \| \\
 &\leq \frac{\epsilon}{3} + \| \mathbf{A}(\bm{\theta}^{\ast}) \| \| \mathbf{B}(\bm{\theta}^{\ast}) - \mathbf{B}(\bm{\theta}_0) \| \| \mathbf{A}(\bar{\bm{\theta}}) \| + \| \mathbf{A}(\bm{\theta}^{\ast}) - \mathbf{A}(\bar{\bm{\theta}}) \| \| \mathbf{B}(\bm-{\theta}_0) \| \| \mathbf{A}(\bar{\bm{\theta}}) \|  \\
 &\leq \frac{\epsilon}{3} + A^{\ast} \frac{\epsilon}{3A^{\ast} A_{bnd}} A_{bnd} + \frac{\epsilon}{3*\max(A^{\ast} B^{\ast}, A_{bnd} B_{bnd})} B_{bnd} A_{bnd} \\
 &\leq \frac{\epsilon}{3} + \frac{\epsilon}{3} + \frac{\epsilon}{3} \\
 &= \epsilon
\end{align}
and the proof is done.  
\end{proof}

\begin{mylemma}\label{anbnconv}
With probability one, for all $\epsilon > 0$ there exists a $\delta > 0$ and an $N$ such that for all $n > N$ and all $\bm{\theta}_0, \bm{\theta}_1, \ldots, \bm{\theta}_k$ with $\| \bm{\theta}_i - \bm{\theta}^{\ast} \| < \delta$ for $i = 0,1,\ldots,k$, $\| \mathbf{A}_n(\bar{\bm{\theta}}) \mathbf{B}_n(\bm{\theta}_0) \mathbf{A}_n(\bar{\bm{\theta}}) - \mathbf{A}(\bm{\theta}^{\ast}) \mathbf{B}(\bm{\theta}^{\ast}) \mathbf{A}(\bm{\theta}^{\ast}) \|< \epsilon$ for $\mathbf{A}_n(\bar{\bm{\theta}})$ defined above.
\end{mylemma}
\begin{proof}
We have by \citep{jennrich1969asymptotic}, Theorem 2, that the almost sure convergence of $\mathbf{A}_n(\bar{\bm{\theta}}) \mathbf{B}_n(\bm{\theta}_0) \mathbf{A}_n(\bar{\bm{\theta}})$ to $\mathbf{A}(\bar{\bm{\theta}}) \mathbf{B}(\bm{\theta}_0) \mathbf{A}(\bar{\bm{\theta}})$ is uniform in $\bm{\theta}$.  

Consequently, pick $\epsilon > 0$.  With probability one, there is an $N$ such that for all $n>N$, and for all $\bm{\theta} \in \bm{\Theta}$, $\| \mathbf{A}_n(\bar{\bm{\theta}}) \mathbf{B}_n(\bm{\theta}_0) \mathbf{A}_n(\bar{\bm{\theta}}) - \mathbf{A}(\bar{\bm{\theta}}) \mathbf{B}(\bm{\theta}_0) \mathbf{A}(\bar{\bm{\theta}}) \| < \epsilon/2$ with $\mathbf{A}_n(\bar{\bm{\theta}})$ as defined above.

For that same $\epsilon > 0$ by Lemma~\ref{abaconv} there is a $\delta > 0$ such that for all $\bm{\theta}_0, \bm{\theta}_1, \ldots, \bm{\theta}_k$ with $\| \bm{\theta}_i - \bm{\theta}^{\ast} \| < \delta$ for $i = 0,1,\ldots,k$, $\| \mathbf{A}(\bar{\bm{\theta}}) \mathbf{B}(\bm{\theta}_0) \mathbf{A}(\bar{\bm{\theta}}) - \mathbf{A}(\bm{\theta}^{\ast}) \mathbf{B}(\bm{\theta}^{\ast}) \mathbf{A}(\bm{\theta}^{\ast}) \|< \epsilon/2$

Combining the two inequalities, we see that with probability one, for any $\epsilon > 0$ there is both $\delta > 0$ and $N$ such that for all $n>N$ and for all $\bm{\theta}_0, \bm{\theta}_1, \ldots, \bm{\theta}_k$ with $\| \bm{\theta}_i - \bm{\theta}^{\ast} \| < \delta$ for $i = 0,1,\ldots,k$, we have $\| \mathbf{A}_n(\bar{\bm{\theta}}) \mathbf{B}_n(\bm{\theta}_0) \mathbf{A}_n(\bar{\bm{\theta}}) - \mathbf{A}(\bm{\theta}^{\ast}) \mathbf{B}(\bm{\theta}^{\ast}) \mathbf{A}(\bm{\theta}^{\ast}) \|< \epsilon$.
\end{proof}

\begin{mylemma} \label{epsmatbd}
Suppose $\mathbf{A}$ is a $p \times p$ matrix, and $\mathbf{x}$ is a $p$-dimensional vector.  Then $(-1)\| \mathbf{A} \| \cdot p^2 \cdot \| \mathbf{x} \|^2 \leq \mathbf{x}^T \mathbf{A} \mathbf{x} \leq \| \mathbf{A} \| \cdot p^2 \cdot \| \mathbf{x} \|^2$
\end{mylemma}
\begin{proof}
We will prove the right-hand inequality.  The left-hand inequality is done similarly.
\begin{align}
\mathbf{x}^T \mathbf{A} \mathbf{x} &= \sum_{i,j} a_{i,j} x_i x_j \\
 &\leq | \sum_{i,j} a_{i,j} x_i x_j | \\
 &\leq \sum_{i,j} | a_{i,j} x_i x_j | \\
 &\leq \sum_{i,j} \underset{i,j}{\max} | a_{i,j} | \cdot | x_i x_j | \\
 &= \underset{i,j}{\max} | a_{i,j} | \left( \sum_{i} |x_i| \right)^2 \\
 &\leq \underset{i,j}{\max} | a_{i,j} | \cdot p^2 \cdot \underset{i}{\max} |x_i|^2 \\
 &\leq \| \mathbf{A} \| \cdot p^2 \cdot \| \mathbf{x} \|^2
\end{align}
\end{proof}

We require that the eigenvalues of all relevant inverse sandwich matrix variations be finite and positive, and specifically bounded away from zero.  First we look at $\mathbf{A}(\bm{\theta}^{\ast}) \mathbf{B}(\bm{\theta}^{\ast})^{-1} \mathbf{A}(\bm{\theta}^{\ast})$.  The standard regularity assumptions in \citep{white1982maximum} state that $\mathbf{A}(\bm{\theta}^{\ast})$ is full rank and that $\mathbf{B}(\bm{\theta}^{\ast})$ is invertible.  That means that $\mathbf{B}(\bm{\theta}^{\ast})$ is positive definite.  Which implies that the asymptotic inverse sandwich, $\mathbf{A}(\bm{\theta}^{\ast}) \mathbf{B}(\bm{\theta}^{\ast})^{-1} \mathbf{A}(\bm{\theta}^{\ast})$, is also positive definite.  This means that the asymptotic inverse sandwich has eigenvalues that are all bounded and positive.  In particular, there is an eigenvalue with smallest magnitude, which we will call $\lambda^{\ast}_{min}$ and an eigenvalue with largest magnitude, which we will call $\lambda^{\ast}_{max}$.  For completeness, we have $0< \lambda^{\ast}_{min} \leq \lambda^{\ast}_{max} < \infty$

Bounding the eigenvalues of $\mathbf{A}_n(\hat{\bm{\theta}}_n) \mathbf{B}_n(\hat{\bm{\theta}}_n)^{-1} \mathbf{A}_n(\hat{\bm{\theta}}_n)$ is now relatively easy.  Since we have $\mathbf{A}_n(\hat{\bm{\theta}}_n) \mathbf{B}_n(\hat{\bm{\theta}}_n)^{-1} \mathbf{A}_n(\hat{\bm{\theta}}_n) \rightarrow_{a.s.} \mathbf{A}(\bm{\theta}^{\ast}) \mathbf{B}(\bm{\theta}^{\ast})^{-1} \mathbf{A}(\bm{\theta}^{\ast})$, we can bound the eigenvalues of $\mathbf{A}_n(\hat{\bm{\theta}}_n) \mathbf{B}_n(\hat{\bm{\theta}}_n)^{-1} \mathbf{A}_n(\hat{\bm{\theta}}_n)$ arbitrarily closely to the eigenvalues of $\mathbf{A}(\bm{\theta}^{\ast}) \mathbf{B}(\bm{\theta}^{\ast})^{-1} \mathbf{A}(\bm{\theta}^{\ast})$.  For any $\epsilon > 0$ with $\epsilon < \lambda^{\ast}_{min}$ we can then bound the range of eigenvalues of $\mathbf{A}_n(\hat{\bm{\theta}}_n) \mathbf{B}_n(\hat{\bm{\theta}}_n)^{-1} \mathbf{A}_n(\hat{\bm{\theta}}_n)$ to be between $\lambda^{\ast}_{min}-\epsilon$ and $\lambda^{\ast}_{max}+\epsilon$, for $n$ sufficiently large enough.

To bound the eigenvalues of $\mathbf{A}_n(\bar{\bm{\theta}}) \mathbf{B}_n(\bm{\theta})^{-1} \mathbf{A}_n(\bar{\bm{\theta}})$, we need another few steps.  These are done with Lemmas~\ref{convexsum}, \ref{abaconv}, and \ref{anbnconv}, which give us the results we need.  Namely that with probability one, for all $\epsilon > 0$ there is both a $\delta > 0$ and $N$ such that for all possible combinations of $\bm{\theta}$ in the ball of radius $\delta$ around $\bm{\theta}^{\ast}$ and for all $n>N$, $\| \mathbf{A}_n(\bar{\bm{\theta}}) \mathbf{B}_n(\bm{\theta}_0)^{-1} \mathbf{A}_n(\bar{\bm{\theta}}) - \mathbf{A}(\bm{\theta}^{\ast}) \mathbf{B}(\bm{\theta}^{\ast})^{-1} \mathbf{A}(\bm{\theta}^{\ast}) \|< \epsilon$.  This means that with a judicious choice of $\epsilon$ ($\epsilon = \lambda^{\ast}_{min}/2$ for example), we can bound the eigenvalues of $\mathbf{A}_n(\bar{\bm{\theta}}) \mathbf{B}_n(\bm{\theta}_0)^{-1} \mathbf{A}_n(\bar{\bm{\theta}})$ to be positive and away from zero using Theorem 5.3 in \citep{zhan2013matrix}.  

\begin{mylemma}\label{bigO}
For any solution, $\breve{\bm{\theta}}_{n,PV} = \hat{\bm{\theta}}_n - \tilde{\bm{\theta}}_{n,PV}$, to the pivot equality, with probability one $\| \tilde{\bm{\theta}}_{n,PV} \| = O(1/\sqrt{n})$.
\end{mylemma}
\begin{proof}
Select $\epsilon$ such that $0 < \epsilon < \lambda^{\ast}_{min}/2$.  With probability one, there is a $\delta > 0$ and an $N_{\epsilon}$ such that for all combinations of $\bm{\theta}_0, \bm{\theta}_1, \ldots, \bm{\theta}_k$ with $\| \bm{\theta}_i - \bm{\theta}^{\ast} \|< \delta, ~ i = 0, 1, \ldots, k$ and all $n>N$ we have $\| \mathbf{A}_n(\bar{\bm{\theta}}) \mathbf{B}_n(\bm{\theta}_0)^{-1} \mathbf{A}_n(\bar{\bm{\theta}}) - \mathbf{A}(\bm{\theta}^{\ast}) \mathbf{B}(\bm{\theta}^{\ast})^{-1} \mathbf{A}(\bm{\theta}^{\ast}) \|< \epsilon < \lambda^{\ast}_{min}/2$.  This implies that the eigenvalues $\lambda_1, \ldots, \lambda_p$ of $\mathbf{A}_n(\bar{\bm{\theta}}) \mathbf{B}_n(\bm{\theta}_0)^{-1} \mathbf{A}_n(\bar{\bm{\theta}})$ satisfy $0 < \lambda^{\ast}_{min}/2 < \lambda_j < \lambda^{\ast}_{max} + \lambda^{\ast}_{min}/2 < \infty, ~ j = 1, \ldots, p$.  Let us simplify this and say that there are bounds $0<\lambda_{lo}$ and $\lambda_{hi}<\infty$ such that $0 < \lambda_{lo} < \lambda_j < \lambda_{hi} < \infty, ~ j = 1, \ldots, p$.

By the almost sure convergence of $\hat{\bm{\theta}}_n$, we have with probability one that there exists an $N_{\delta}$ such that for all $n> N_{\delta}$, $\| \hat{\bm{\theta}}_n - \bm{\theta}^{\ast} \| < \delta/2$, where $\delta$ is defined in the previous paragraph.  This means that with probability one, the $\delta/2$ ball around $\hat{\bm{\theta}}_n$ is contained within the $\delta$ ball around $\bm{\theta}^{\ast}$.  We can take $N_{max} = \max(N_{\epsilon}, N_{\delta})$ so that all of the results in this and the previous paragraph hold.  

Consider the case $n> N_{max}$, and select a point on the boundary of the $\delta/2$ ball around $\hat{\bm{\theta}}_n$.  Let's refer to this point as $\bm{\theta}_0 = \hat{\bm{\theta}}_n - \tilde{\bm{\theta}}_n$ because we want to emphasize the fact that $\| \tilde{\bm{\theta}}_n \| = \delta/2$.  We can plug this point in for $\bm{\theta}$ in the pivotal quadratic form.  Recall we also use it and $\hat{\bm{\theta}}_n$ and as the points to define $\mathbf{A}_n(\bar{\bm{\theta}})$.  Using results on the bounds of quadratic forms of symmetric positive definite matrices we have
\begin{align}
\lambda_{lo} \| \tilde{\bm{\theta}}_n \|^2 &\leq (\tilde{\bm{\theta}}_n)^T \mathbf{A}_n(\bar{\bm{\theta}}) \mathbf{B}_n(\bm{\theta}_0)^{-1} \mathbf{A}_n(\bar{\bm{\theta}}) (\tilde{\bm{\theta}}_n) \leq \lambda_{hi} \| \tilde{\bm{\theta}}_n \|^2 \\
\lambda_{lo} \frac{\delta^2}{4} &\leq (\tilde{\bm{\theta}}_n)^T \mathbf{A}_n(\bar{\bm{\theta}}) \mathbf{B}_n(\bm{\theta}_0)^{-1} \mathbf{A}_n(\bar{\bm{\theta}}) (\tilde{\bm{\theta}}_n) \leq \lambda_{hi} \frac{\delta^2}{4}. \label{bounds}
\end{align}
which is true regardless of the direction of $\tilde{\bm{\theta}}_n$, hence it is uniformly true in all directions away from $\hat{\bm{\theta}}_n$.  

We now compare (\ref{bounds}) with the main pivot inequality in (\ref{pivotmatineq}).  As $n$ increases, we see that the lower bound in (\ref{bounds}) is a constant.  However as $n$ increases in (\ref{pivotmatineq}), it will eventually be the case that $\chi^2_{p, 1-\alpha} / n < \lambda_{lo} \frac{\delta^2}{4}$, namely when $n > \frac{4 \chi^2_{p, 1-\alpha}}{\lambda_{lo} \delta^2}$.  Since $\hat{\bm{\theta}}_n$ makes the pivotal inequality equal to zero, and every point on the $\delta/2$ ball makes the quadratic form evaluate to more than $\chi^2_{p, 1-\alpha} / n$, we can invoke the intermediate value theorem to conclude that all solutions to the pivotal inequality will lie within the $\delta/2$ ball around $\hat{\bm{\theta}}_n$, and that they do exist.  We can now take $N = \max(\frac{4 \chi^2_{p, 1-\alpha}}{\lambda_{lo} \delta^2}, N_{max})$.  

Define $N$ and $\delta$ as we have done.  We consider a solution $\breve{\bm{\theta}}_{n,PV} = \hat{\bm{\theta}}_n - \tilde{\bm{\theta}}_{n,PV}$ to the pivotal equality in an arbitrary direction from $\hat{\bm{\theta}}$, and specifically look at the equality
\begin{align}
(\tilde{\bm{\theta}}_{n,PV})^T \mathbf{A}_n(\bar{\bm{\theta}}) \mathbf{B}_n(\breve{\bm{\theta}}_{n,PV})^{-1} \mathbf{A}_n(\bar{\bm{\theta}}) (\tilde{\bm{\theta}}_{n,PV}) &= \frac{\chi^2_{p, 1-\alpha}}{n}.
\end{align}
Given the eigenvalue-based bounds above and using knowledge of quadratic forms, we can conclude
\begin{align}
\lambda_{lo} \| \tilde{\bm{\theta}}_{n,PV} \|^2 &\leq (\tilde{\bm{\theta}}_{n,PV})^T \mathbf{A}_n(\bar{\bm{\theta}}) \mathbf{B}_n(\breve{\bm{\theta}}_{n,PV})^{-1} \mathbf{A}_n(\bar{\bm{\theta}}) (\tilde{\bm{\theta}}_{n,PV})  \leq \lambda_{hi} \| \tilde{\bm{\theta}}_{n,PV} \|^2 \\
\lambda_{lo} \| \tilde{\bm{\theta}}_{n,PV} \|^2 &\leq \frac{\chi^2_{p, 1-\alpha}}{n} \leq \lambda_{hi} \| \tilde{\bm{\theta}}_{n,PV} \|^2
\end{align}
and solving for $\| \tilde{\bm{\theta}}_{n,PV} \|$ we have
\begin{align}
\sqrt{\frac{\chi^2_{p, 1-\alpha}}{n \lambda_{hi}}} &\leq \| \tilde{\bm{\theta}}_{n,PV} \| \leq \sqrt{\frac{\chi^2_{p, 1-\alpha}}{n \lambda_{lo}}}
\end{align}
which gives us the result that $\| \tilde{\bm{\theta}}_{n,PV} \| = O(1/\sqrt{n})$ with probability one.
\end{proof}

\textbf{Proof of Theorem \ref{multiasythm}}
\begin{proof}
Consider solutions $\breve{\bm{\theta}}_{n,PV}$ and $\breve{\bm{\theta}}_{n,SW}$ to the pivot equality and sandwich equality, respectively and both are in the same direction from $\hat{\bm{\theta}}_n$.  First, rewrite them as
\begin{align}
\breve{\bm{\theta}}_{n,PV} &= \hat{\bm{\theta}}_n - \tilde{\bm{\theta}}_{n,PV} \\
\breve{\bm{\theta}}_{n,SW} &= \hat{\bm{\theta}}_n - \tilde{\bm{\theta}}_{n,SW}
\end{align}
where we emphasize again that $\tilde{\bm{\theta}}_{n,PV}$ and $\tilde{\bm{\theta}}_{n,SW}$ are vectors having the same direction, but with possibly different magnitudes.  

We have that $\| \tilde{\bm{\theta}}_{n,PV} \|$ is $O(1/\sqrt{n})$ with probability one by Lemma \ref{bigO}.  This also tells us that with probability one as $n\rightarrow \infty$, $\breve{\bm{\theta}}_{n,PV}$ will become arbitrarily close to $\hat{\bm{\theta}}_n$ and specifically, there is a constant $A$ such that for sufficiently large $n$, $\| \tilde{\bm{\theta}}_{n,PV} \| \leq A/ \sqrt{n}$.  We also have that $\| \mathbf{A}_n(\hat{\bm{\theta}}_n) \mathbf{B}_n(\hat{\bm{\theta}}_n)^{-1} \mathbf{A}_n(\hat{\bm{\theta}}_n) - \mathbf{A}_n(\bar{\bm{\theta}}) \mathbf{B}_n(\breve{\bm{\theta}}_{n,PV})^{-1} \mathbf{A}_n(\bar{\bm{\theta}}) \| \rightarrow_{a.s.} 0$ by the continuity of the matrix functions $\mathbf{A}_n(\bm{\theta})$ and $\mathbf{B}_n(\bm{\theta})^{-1}$.  Pick $\zeta > 0$ small enough so that $\zeta < \chi^2_{p,1-\alpha}$.  With probability one there is an $M$ such that for all $n>M$, $\| \mathbf{A}_n(\hat{\bm{\theta}}_n) \mathbf{B}_n(\hat{\bm{\theta}}_n)^{-1} \mathbf{A}_n(\hat{\bm{\theta}}_n) - \mathbf{A}_n(\bar{\bm{\theta}}_n) \mathbf{B}_n(\breve{\bm{\theta}}_{n,PV})^{-1} \mathbf{A}_n(\bar{\bm{\theta}}_n) \| < \frac{\zeta}{p^2 A^2}$.

Now take a solution, $\breve{\bm{\theta}}_{n,PV} = \hat{\bm{\theta}}_n - \tilde{\bm{\theta}}_{n,PV}$, to the pivot equality and plug it into the sandwich equality
\begin{align}
& (\hat{\bm{\theta}}_n - \breve{\bm{\theta}}_{n,PV})^T \mathbf{A}_n(\hat{\bm{\theta}}) \mathbf{B}_n(\hat{\bm{\theta}})^{-1} \mathbf{A}_n(\hat{\bm{\theta}}) (\hat{\bm{\theta}}_n - \breve{\bm{\theta}}_{n,PV}) =\\
& (\tilde{\bm{\theta}}_{n,PV})^T \mathbf{A}_n(\hat{\bm{\theta}}_n) \mathbf{B}_n(\hat{\bm{\theta}}_n)^{-1} \mathbf{A}_n(\hat{\bm{\theta}}_n) (\tilde{\bm{\theta}}_{n,PV}) = \label{swpivsoln}\\
& (\tilde{\bm{\theta}}_{n,PV})^T [ \mathbf{A}_n(\hat{\bm{\theta}}_n) \mathbf{B}_n(\hat{\bm{\theta}}_n)^{-1} \mathbf{A}_n(\hat{\bm{\theta}}_n) - \\
&\qquad \mathbf{A}_n(\bar{\bm{\theta}}) \mathbf{B}_n(\breve{\bm{\theta}}_{n,PV})^{-1} \mathbf{A}_n(\bar{\bm{\theta}}) + \mathbf{A}_n(\bar{\bm{\theta}}) \mathbf{B}_n(\breve{\bm{\theta}}_{n,PV})^{-1} \mathbf{A}_n(\bar{\bm{\theta}}) ] (\tilde{\bm{\theta}}_{n,PV}) = \\
& (\tilde{\bm{\theta}}_{n,PV})^T \mathbf{A}_n(\bar{\bm{\theta}}) \mathbf{B}_n(\breve{\bm{\theta}}_{n,PV})^{-1} \mathbf{A}_n(\bar{\bm{\theta}}) (\tilde{\bm{\theta}}_{n,PV}) + \\
&\qquad (\tilde{\bm{\theta}}_{n,PV})^T [\mathbf{A}_n(\hat{\bm{\theta}}_n) \mathbf{B}_n(\hat{\bm{\theta}}_n)^{-1} \mathbf{A}_n(\hat{\bm{\theta}}_n) - \mathbf{A}_n(\bar{\bm{\theta}}) \mathbf{B}_n(\breve{\bm{\theta}}_{n,PV})^{-1} \mathbf{A}_n(\bar{\bm{\theta}})] (\tilde{\bm{\theta}}_{n,PV}) = \\
& \frac{\chi^2_{p, 1-\alpha}}{n} + (\tilde{\bm{\theta}}_{n,PV})^T [\mathbf{A}_n(\hat{\bm{\theta}}_n) \mathbf{B}_n(\hat{\bm{\theta}}_n)^{-1} \mathbf{A}_n(\hat{\bm{\theta}}_n) - \mathbf{A}_n(\bar{\bm{\theta}}) \mathbf{B}_n(\breve{\bm{\theta}}_{n,PV})^{-1} \mathbf{A}_n(\bar{\bm{\theta}})] (\tilde{\bm{\theta}}_{n,PV}). \label{line32}
\end{align}
By the previous part of this paragraph we have that
\begin{align}
-\frac{A^2}{n} p^2 \frac{\zeta}{p^2 A^2} &\leq (\tilde{\bm{\theta}}_{n,PV})^T \left[\mathbf{A}_n(\hat{\bm{\theta}}_n) \mathbf{B}_n(\hat{\bm{\theta}}_n)^{-1} \mathbf{A}_n(\hat{\bm{\theta}}_n)  \right. \nonumber \\
&\qquad \left. - \mathbf{A}_n(\bar{\bm{\theta}}) \mathbf{B}_n(\breve{\bm{\theta}}_{n,PV})^{-1} \mathbf{A}_n(\bar{\bm{\theta}}) \right] (\tilde{\bm{\theta}}_{n,PV}) \leq \frac{A^2}{n} p^2 \frac{\zeta}{p^2 A^2} \nonumber \\
-\frac{\zeta}{n} &\leq (\tilde{\bm{\theta}}_{n,PV})^T [\mathbf{A}_n(\hat{\bm{\theta}}_n) \mathbf{B}_n(\hat{\bm{\theta}}_n)^{-1} \mathbf{A}_n(\hat{\bm{\theta}}_n) - \mathbf{A}_n(\bar{\bm{\theta}}) \mathbf{B}_n(\breve{\bm{\theta}}_{n,PV})^{-1} \mathbf{A}_n(\bar{\bm{\theta}})] (\tilde{\bm{\theta}}_{n,PV}) \leq \frac{\zeta}{n}  \label{twosidinequ}
\end{align}
and plugging Line (\ref{twosidinequ}) into Line (\ref{line32}) we have
\begin{align}
\frac{\chi^2_{p, 1-\alpha} - \zeta}{n} &\leq (\tilde{\bm{\theta}}_{n,PV})^T \mathbf{A}_n(\hat{\bm{\theta}}_n) \mathbf{B}_n(\hat{\bm{\theta}})^{-1} \mathbf{A}_n(\hat{\bm{\theta}}_n) (\tilde{\bm{\theta}}_{n,PV}) \leq \frac{\chi^2_{p, 1-\alpha} + \zeta}{n}
\end{align}
Taking each side of the above inequality separately, we see
\begin{align}
\frac{\chi^2_{p, 1-\alpha}}{\chi^2_{p, 1-\alpha} + \zeta} (\tilde{\bm{\theta}}_{n,PV})^T \mathbf{A}_n(\hat{\bm{\theta}}_n) \mathbf{B}_n(\hat{\bm{\theta}}_n)^{-1} \mathbf{A}_n(\hat{\bm{\theta}}_n) (\tilde{\bm{\theta}}_{n,PV})  &< \frac{\chi^2_{p, 1-\alpha}}{n} \\
\frac{\chi^2_{p, 1-\alpha}}{\chi^2_{p, 1-\alpha} - \zeta} (\tilde{\bm{\theta}}_{n,PV})^T \mathbf{A}_n(\hat{\bm{\theta}}_n) \mathbf{B}_n(\hat{\bm{\theta}}_n)^{-1} \mathbf{A}_n(\hat{\bm{\theta}}_n) (\tilde{\bm{\theta}}_{n,PV})  &> \frac{\chi^2_{p, 1-\alpha}}{n} 
\end{align}
Now we can rewrite $\frac{\chi^2_{p, 1-\alpha}}{\chi^2_{p, 1-\alpha} + \zeta}$ as $1 - \frac{\zeta}{\chi^2_{p, 1-\alpha} + \zeta} = 1 - \xi_+$.  Similarly, $\frac{\chi^2_{p, 1-\alpha}}{\chi^2_{p, 1-\alpha} - \zeta} = 1+ \frac{\zeta}{\chi^2_{p, 1-\alpha} - \zeta} = 1+ \xi_-$.  Note in particular that $\zeta$ and $\xi_-$ are in one-to-one correspondence, and as $\zeta \rightarrow 0$, $\xi_- \rightarrow 0$.  This gives us
\begin{align}
(1-\xi_+) (\tilde{\bm{\theta}}_{n,PV})^T \mathbf{A}_n(\hat{\bm{\theta}}_n) \mathbf{B}_n(\hat{\bm{\theta}}_n)^{-1} \mathbf{A}_n(\hat{\bm{\theta}}_n) (\tilde{\bm{\theta}}_{n,PV})  &< \frac{\chi^2_{p, 1-\alpha}}{n} \\
(1+\xi_-) (\tilde{\bm{\theta}}_{n,PV})^T \mathbf{A}_n(\hat{\bm{\theta}}_n) \mathbf{B}_n(\hat{\bm{\theta}}_n)^{-1} \mathbf{A}_n(\hat{\bm{\theta}}_n) (\tilde{\bm{\theta}}_{n,PV})  &> \frac{\chi^2_{p, 1-\alpha}}{n} .
\end{align}
Since $\xi_- > \xi_+$ and $1- \xi_- < 1- \xi_+$ we can rewrite the first inequality to get
\begin{align}
(1-\xi_-) (\tilde{\bm{\theta}}_{n,PV})^T \mathbf{A}_n(\hat{\bm{\theta}}_n) \mathbf{B}_n(\hat{\bm{\theta}}_n)^{-1} \mathbf{A}_n(\hat{\bm{\theta}}_n) (\tilde{\bm{\theta}}_{n,PV})  &< \frac{\chi^2_{p, 1-\alpha}}{n} \\
(1+\xi_-) (\tilde{\bm{\theta}}_{n,PV})^T \mathbf{A}_n(\hat{\bm{\theta}}_n) \mathbf{B}_n(\hat{\bm{\theta}}_n)^{-1} \mathbf{A}_n(\hat{\bm{\theta}}_n) (\tilde{\bm{\theta}}_{n,PV})  &> \frac{\chi^2_{p, 1-\alpha}}{n} .
\end{align}
We have that $\hat{\bm{\theta}}_n - (\sqrt{1-\xi_-})\tilde{\bm{\theta}}_{n,PV}$ satisfies the plug-in sandwich quadratic form inequality but $\hat{\bm{\theta}}_n - (\sqrt{1+\xi_-})\tilde{\bm{\theta}}_{n,PV}$ does not.  By the continuity of the plug-in sandwich quadratic form, the solution to the sandwich equality, i.e. $\breve{\bm{\theta}}_{n,SW} = \hat{\bm{\theta}}_n - \tilde{\bm{\theta}}_{n,SW}$, must lie on the line segment between these two points.  

Once we have selected $\zeta$ and consequently $\xi_-$ as described above, $1-\sqrt{1-\xi_-} > \sqrt{1+\xi_-} - 1$.  This means that $\hat{\bm{\theta}}_n - (\sqrt{1-\xi_-})\tilde{\bm{\theta}}_{n,PV}$ is the furthest point from $\hat{\bm{\theta}}_n - \tilde{\bm{\theta}}_{n,PV}$ in the interval $(\hat{\bm{\theta}}_n - (\sqrt{1-\xi_-})\tilde{\bm{\theta}}_{n,PV}, \, \hat{\bm{\theta}}_n - (\sqrt{1+\xi_-})\tilde{\bm{\theta}}_{n,PV})$.  Putting this all together, we have that for all $1 > \xi_- > 0$ with probability one there is an $M$ such that for all $n> M$
\begin{align}
\sqrt{n} \| (\hat{\bm{\theta}}_n - \tilde{\bm{\theta}}_{n,PV}) - (\hat{\bm{\theta}}_n - \tilde{\bm{\theta}}_{n,SW}) \| &= \sqrt{n} \| \tilde{\bm{\theta}}_{n,SW} - \tilde{\bm{\theta}}_{n,PV} \| \\
 &\leq \| (\sqrt{1-\xi_-})\tilde{\bm{\theta}}_{n,PV} - \tilde{\bm{\theta}}_{n,PV} \| \\
 &= \sqrt{n} \| \tilde{\bm{\theta}}_{n,PV} \| \mid 1-\sqrt{1-\xi_-} \mid \\
 &\leq \sqrt{n} \sqrt{\frac{\chi^2_{p, 1-\alpha}}{n \lambda_{lo}}} (\mid 1-\sqrt{1-\xi_-} \mid ) \\
 &= \sqrt{\frac{\chi^2_{p, 1-\alpha}}{\lambda_{lo}}} (\mid 1-\sqrt{1-\xi_-} \mid)
\end{align}
but $\zeta$ was chosen to be arbitrarily small, which gives us that $\xi_-$ is also arbitrarily small.  Hence we have that $\sqrt{n} \| (\hat{\bm{\theta}}_n - \tilde{\bm{\theta}}_{n,PV}) - (\hat{\bm{\theta}}_n - \tilde{\bm{\theta}}_{n,SW}) \| \rightarrow_{a.s.} 0$ as $n\rightarrow \infty$, as desired.
\end{proof}

\bibliography{PivotPaper}
\end{doublespace}
\end{document}